\documentclass[a4paper,11pt]{article}
\usepackage{amsfonts,amsthm,amssymb}
\usepackage{authblk}
\newcommand\COMP{\hbox{C\kern -.58em {\raise .54ex \hbox{$\scriptscriptstyle |$}}
\kern-.55em {\raise .53ex \hbox{$\scriptscriptstyle |$}} }}
\newcommand\NN{\hbox{I\kern-.2em\hbox{N}}}
\newcommand\RR{\hbox{I\kern-.2em\hbox{R}}}
\newcommand\sRR{{\it \hbox{I\kern-.2em\hbox{R}}}}
\newcommand\QQ{\hbox{I\kern-.53em\hbox{Q}}}
\newcommand\PP{\hbox{I\kern-.53em\hbox{P}}}
\newcommand\EE{\hbox{I\kern-.53em\hbox{E}}}
\newcommand\ZZ{{{\rm Z}\kern-.28em{\rm Z}}}
\newcommand\be{\begin{equation}}
\newcommand\ee{\end{equation}}
\newtheorem{theorem}{Theorem}[section]

\newtheorem{proposition}[theorem]{Proposition}
\newtheorem{remark}[theorem]{Remark}

\newtheorem{lemma}[theorem]{Lemma}

\newtheorem{definition}[theorem]{Definitions}
\newtheorem{corollary}[theorem]{Corollary}

\def \Rbrack {[\![}
\def \Lbrack {]\!]}

\def \Rbrack {[\![}
\def \Lbrack {]\!]}

\newcommand{\is}{\centerdot}

\begin{document}
\title{Non-Arbitrage Under Additional Information for Thin Semimartingale Models\thanks{The research of Tahir Choulli  and Jun Deng is supported financially by the
Natural Sciences and Engineering Research Council of Canada,
through Grant G121210818. The research of Anna Aksamit and Monique Jeanblanc is supported
by Chaire Markets in transition, French Banking Federation.}}

\author[2]{Anna  Aksamit}
\author[1]{Tahir Choulli\thanks{corresponding author, {Email: tchoulli@ualberta.ca} }}
\author[3]{Jun Deng}
\author[2]{Monique Jeanblanc}
\affil[1]{Mathematical and Statistical Sciences Depart., University of Alberta, Edmonton, Canada}
\affil[3]{School of Banking and Finance,
University of International Business and Economics,
 Beijing, China}
\affil[2]{Laboratoire Analyse et Probabilit\'es,
Universit\'e d'Evry  Val d'Essonne,  Evry, France \newline  \newline
\textbf{This paper develops the part of thin and single jump processes mentioned in our earlier version: ''Non-arbitrage up to random horizon and after honest times for semimartingale models'',  Available at: http://arxiv.org/abs/1310.1142v1}}

%\author{ Tahir Choulli\thanks{corresponding author, {Email: tchoulli@ualberta.ca} },   Anna  Aksamit,  Jun Deng, and Monique Jeanblanc\\
%Mathematical and Statistical Sciences Depart.,
%University of Alberta, Edmonton, Canada \\
%Laboratoire Analyse et Probabilit\'es,
%Universit\'e d'Evry  Val d'Essonne,  Evry, France
%}

%\date{}

\maketitle

\begin{abstract}
This paper completes the two studies undertaken in \cite{aksamit/choulli/deng/jeanblanc2} and \cite{aksamit/choulli/deng/jeanblanc3}, where the authors quantify the impact of a random time on the No-Unbounded-Risk-with-Bounded-Profit concept (called NUPBR hereafter) when the stock price processes are quasi-left-continuous (do not jump on predictable stopping times). Herein, we focus on  the NUPBR for semimartingales models that live on thin predictable sets only and the progressive enlargement with a random time. For this flow of information, we explain how far the NUPBR property is affected when one stops the model by an arbitrary random time or when one incorporates fully an honest time into the model. This also generalizes \cite{choulli/deng} to the case when the jump times are not ordered in anyway. Furthermore, for the current context, we show how to construct explicitly  local martingale deflator under the bigger filtration from those of the smaller filtration. \end{abstract}

%\begin{keyword}[class=MSC]
%\kwd[Primary ]{}
%\kwd{}
%\kwd[; secondary ]{}
%\end{keyword}

%\begin{keyword}
%\kwd{}
%\kwd{}
%\end{keyword}

% AOS,AOAS: If there are supplements please fill:
%\begin{supplement}[id=suppA]
%  \sname{Supplement A}
%  \stitle{Title}
%  \slink[doi]{10.1214/00-AOASXXXXSUPP}
%  \sdatatype{.pdf}"
%  \sdescription{Some text}
%\end{supplement}

%\newpage
\section{Introduction} We consider a stochastic basis $(\Omega, {\cal G},  {\mathbb  F }=({\cal F}_t)_{t\geq 0},  P)$,  where
${\mathbb  F }$ is a filtration
satisfying the usual hypotheses (i.e., right continuity and completeness), and $ {\cal F}_\infty
 \subseteq {\cal {G}}$. Financially speaking, the filtration $\mathbb F$ represents the flow of public information through time. On this basis, we consider an arbitrary but fixed $d$-dimensional c\`adl\`ag semimartingale, $S$, which represents the discounted price processes of $d$-stocks, while the riskless asset's price is assumed to be constant.
 Beside  the initial  model
 $\left(\Omega,{\cal G}, \mathbb F, P, S\right)$, we  consider a random time $\tau$, i.e. a non-negative ${\cal G}$-measurable random variable. At the practical level, this random time can model the death time, the default time of a firm, or any occurrence time of an event that might affect the market in some way. The main goal of this paper lies in discussing whether the new model $(S,\mathbb F, \tau)$ is arbitrage free or not. To address this question rigourously, we need to specify the non-arbitrage concept adopted herein on the one hand, as arbitrage in continuous time has competing definitions. On the other hand, one need to model the flow of information that catch both the flow $\mathbb F$ and the information represented by $\tau$. To this random time, we associate  the process $D$ and the
filtration $\mathbb G$ given by
\begin{equation}\label{DandfiltrationG}
D:= I_{\Rbrack\tau,+\infty\Rbrack},\ \ \ \mathbb G=\left({\cal G}_t\right)_{t\geq 0},\ \ \ {\cal G}_t =
\displaystyle\bigcap_{s>t}\Bigl({\cal F}_s\vee \sigma(D_u, u\leq s)\Bigr).\end{equation}
The filtration $\mathbb G$ is the smallest
right-continuous filtration which contains ${\mathbb  F }$ and makes $\tau$ a stopping time. In the probabilistic literature, $\mathbb G$ is called the progressive enlargement of $\mathbb F$ with $\tau$. To define mathematically the non-arbitrage condition, we need to define some notations that will be useful throughout the paper.
\subsection{Some General Notations and Definitions} Throughout the paper, $\mathbb H$ denotes a filtration satisfying
the usual hypotheses and $Q$ a probability measure  on the
filtered probability space $(\Omega, \mathbb H)$. The set of
martingales for the filtration $\mathbb H$ under $Q$ is denoted by
${\cal M}(\mathbb H, Q)$. When $Q=P$, we simply denote ${\cal
M}(\mathbb H)$. As usual, ${\cal A}^+(\mathbb H)$ denotes the set
of increasing,
right-continuous, $\mathbb H$-adapted and integrable processes.\\
\noindent If ${\cal C}(\mathbb H)$ is a class of $\mathbb H$-adapted processes,
 we denote by ${\cal C}_0(\mathbb H)$ the set of processes $X\in {\cal C}(\mathbb H)$ with $X_0=0$, and
 by ${\cal C}_{loc}(\mathbb H)$
 the set  of processes $X$
 such that there exists a sequence $(T_n)_{n\geq 1}$ of $\mathbb H$-stopping times  that increases to $+\infty$ and the
 stopped processes $X^{T_n}$ belong
to ${\cal C}(\mathbb H)$. We put $ {\cal C}_{0,loc}(\mathbb H)={\cal
C}_0(\mathbb H)\cap{{\cal
C}}_{loc}(\mathbb H)$.\\
\noindent For a process $K$ with $\mathbb H$-locally integrable variation, we denote by $K^{o,\mathbb H}$  its dual optional projection. The dual predictable projection of $K$ (also called the $\mathbb H$-dual predictable projection) is denoted   $K^{p,\mathbb H}$. For a process $X$, we denote $^{o, \mathbb H \!} X$ (resp.$\,^{p, \mathbb H\!} X$ ) its optional (resp. predictable) projection with respect to $\mathbb H$.\\
\noindent For an $\mathbb H$- semi-martingale $Y$, the set $L(Y, \mathbb H)$ is the set of $\mathbb H$ predictable processes
integrable w.r.t. $Y$ and  for $H\in L(Y, \mathbb H )$, we denote $H\centerdot  Y_t:= \int_0^t H_sdY_s$.\\
As usual, for a process $X$ and a random time $\vartheta$, we
denote by $X^\vartheta$ the stopped process. To distinguish the
effect of filtration, we will denote $\langle ., .\rangle^{\mathbb
F},$ or $\langle ., . \rangle^{\mathbb G}$
 the sharp bracket (predictable covariation process) calculated in the filtration ${\mathbb F}$ or  ${\mathbb G},$ if confusion may rise. We recall that, for general semi-martingales $X$ and $Y$, the
sharp bracket is (if it exists) the dual predictable projection of
the covariation process $[X,Y]$.\\

\noindent We recall
the definition of thin processes/sets for the reader's convenience
\begin{definition} A set $A\subset \Omega\times [0,\infty[$ is thin if, for all $\omega\in \Omega$,
the set $A(\omega)$ is countable.  A process $X$ is called  thin
if there exists a sequence of random variables $\xi_n$ and an
increasinq sequence of random times $T_n$ such that
$X_t=\sum_{n=1}^\infty \xi_n I_{\Rbrack T_n,\infty\Rbrack} $. Its
paths vary on a thin set only, and hence  $$X= I_{ \cup_{n=1}^\infty\Rbrack
T_n\Lbrack} \centerdot X= \sum_{n=1}^\infty I_{ \Rbrack
T_n\Lbrack} \centerdot X= \sum_{n=1}^\infty I_{ \Rbrack
T_n\Lbrack}  \Delta  X_{ T_n} .$$\end{definition}

\subsection{The non-arbitrage concept }
\noindent We introduce the non-arbitrage notion that will be addressed in this paper.

\begin{definition}\label{DefinitionofNUPBR} An $\mathbb H$-semimartingale
$X$ satisfies the {\it No-Unbounded-Profit-with-Bounded-Risk}
condition under $(\mathbb H,Q)$ (called NUPBR$(\mathbb H,Q)$ hereafter) if for any $T\in (0,+\infty)$ the set
$$\label{boundedset}
{\cal K}_{T} (X,\mathbb H):=\displaystyle \Bigl\{(H\centerdot
 S)_T \ \big|\ \
\ H \in L(X,\mathbb H)\ \mbox{and}\ H\is X\geq -1\ \Bigr\}$$ is
bounded in probability under $Q$. When $Q\sim P$, we simply write,
with an abuse of language, $X$ satisfies NUPBR$(\mathbb H)$.
\end{definition}

\noindent This definition was given in \cite{aksamit/choulli/deng/jeanblanc2}, together with the following .

\begin{proposition}\label{charaterisationofNUPBRloc} Let $X$ be an $\mathbb H$-semimartingale. Then the
 following assertions are equivalent.\\
{\rm{(a)}}  $X$ satisfies NUPBR$(\mathbb H)$.\\
  {\rm{(b)}} There exist a positive $\mathbb H$-local martingale, $Y$ and an
  $\mathbb H$-predictable process $\theta$ satisfying $0<\theta\leq 1$ and $Y(\theta\is X)$
   is a local martingale.
\end{proposition}

\noindent For any $\mathbb H$-semimartingale $X$, the local martingales fulfilling the assertion (b) of Proposition \ref{charaterisationofNUPBRloc} are
 called   $\sigma$-martingale densities for $X$. The set of these $\sigma$-martingale densities will be denoted throughout the paper by
\begin{equation}\label{LFS}
{\cal L}({\mathbb H},X):=\left\{ Y\in {\cal M}_{loc}(\mathbb
H)\big|\ Y>0,\  \exists \theta \in {\cal {P}}(\mathbb H),\,
0<\theta\leq 1,\
 Y(\theta \centerdot X)\in {\cal M}_{loc}(\mathbb
H) \right\}\end{equation} where, as usual, ${\cal {P}}(\mathbb H)$
stands for the predictable $\sigma$-field on $\Omega \times
[0,\infty)$ and by abuse of notation $\theta\in  {\cal
{P}}(\mathbb H)$ means that $\theta$ is  ${\cal {P}}(\mathbb
H)$-measurable.
  We state, without proof, an obvious lemma.
\begin{lemma}\label{LY} For any $\mathbb H$-semimartingale $X$ and any $Y\in {\cal L}({\mathbb H},X)$, one has
$\ \  ^{p,\mathbb H} (Y  \vert \Delta X \vert )<\infty$ and  $
^{p,\mathbb H}(Y \Delta X
  )=0$.
\end{lemma}

%%%%%%%%%%%%%%%%%%%%%%%%%%%%%%%%%%%%%%%%%%%%%%%%%%%%%%%%%%%%%%%%%%%%%%%%%%%%%%%%%%
%%%%%%%%%%%%%%%%%%%%%%%%%%%%%%%%%%%%%%%%%%%%%%%%%%%%%%%%%%%%%%%%%%%%%%%%%%%%%%%%%%%%%

\noindent Below, we  state a result that was proved in \cite{aksamit/choulli/deng/jeanblanc2}, and will be frequently used throughout the paper.
\begin{proposition}\label{NUPBRLocalization}
Let $X$ be an $\mathbb H$ adapted   process. Then, the following assertions are equivalent.\\
{\rm{(a)}}  There exists a sequence   $(T_n)_{n\geq 1}$ of
$\mathbb H$-stopping times that increases to $+\infty$, such that
for each $n\geq 1$, there exists a probability $Q_n$ on $(\Omega,
{\mathbb H}_{T_n})$ such that $Q_n\sim P$ and $X^{T_n}$ satisfies
 NUPBR$(\mathbb H)$ under $Q_n$.\\
{\rm{(b)}}  $X$ satisfies   NUPBR$(\mathbb H)$.\\
{\rm{(c)}} There exists an $\mathbb H$-predictable process $\phi$,
such that $0<\phi\leq 1$ and  $(\phi\is X)$ satisfies
NUPBR$(\mathbb H)$.
\end{proposition}

We end this section with a simple but useful result for
predictable process with finite variation.
\begin{lemma}\label{NUPBRforPredictableProcesses}
Let $X$ be an $\mathbb H$-predictable process with finite variation. Then $X$ satisfies NUPBR$(\mathbb H)$ if and only if $X\equiv X_0$ (i.e. the process $X$ is constant).\end{lemma}

\subsection{Our Achievements}
Given the modeling of the new flow of the information, our main goal becomes whether $(S,\mathbb G)$ satisfies the NUPBR or not when $S$ is an $\mathbb F$-semimartingale. Precisely, we characterise the pair of initial market and the random time $(S,\tau)$ for  which the new market $(S,\mathbb G)$ fulfills the NUPBR. This problem was addressed in \cite{aksamit/choulli/deng/jeanblanc2} and \cite{aksamit/choulli/deng/jeanblanc3} for the parts $(S^{\tau},\mathbb G)$ and $(S-S^{\tau},\mathbb G)$ respectively when $S$ is a quasi-left-continuous process. Thus, the case of thin $\mathbb F$-semimartingale with predictable jumps is not covered in these works. The case of discrete time market with finite horizon is presented in \cite{choulli/deng}.  Hence, the main objective of this work lies in deriving results on the NUPBR for thin processes under additional information generated by a random time. It is important to mention that this work complies the other parts towards understanding the effect of extra information on the NUPBR for general semimartingales. This can be seen by recalling that for an $\mathbb H$-semimartingale, $X$, we
associate a sequence of $\mathbb H$-predictable stopping times
$(T^X_n)_{n\geq 1}$  that exhaust the accessible jump times of
$X$, and put $\Gamma_X:=\bigcup_{n=1}^{\infty} \Rbrack
T_n^X\Lbrack$. Then, we can decompose $X$ as follows.

\begin{equation}\label{Sdecomposition}
X=X^{(qc)}+X^{(a)},\ X^{(a)}:=I_{\Gamma_X}\is X,\
X^{(qc)}:=X-X^{(a)}.\end{equation} The process $X^{(a)}$ (the accessible
part of $X$) is a thin process with  predictable jumps only,
 while $X^{(qc)}$ is a $\mathbb H$-quasi-left-continuous process (the quasi-left-continuous part
 of $X$).

\begin{lemma}\label{NUPBRdecomposed}
 Let $X$ be an $\mathbb{H}$-semimartingale. Then $X$ satisfies NUPBR$(\mathbb{H})$ if and only
 if $X^{(a)}$ and $X^{(qc)}$ satisfy NUPBR$(\mathbb{H})$.
\end{lemma}

\begin{proof}
  Thanks to Proposition \ref{charaterisationofNUPBRloc}, $X$ satisfies NUPBR$(\mathbb{H})$ if
  and only if there exist an $\mathbb{H}$-predictable real-valued process $\phi>0$ and a
  positive $\mathbb{H}$-local martingale $Y$ such that $Y(\phi\centerdot X)$ is an $\mathbb{H}$-
  local martingale. Then, it is obvious that  $Y(\phi I_{\Gamma_X}\centerdot X)$ and
  $Y(\phi I_{{\Gamma_X}^c}\centerdot X)$ are both $\mathbb{H}$-local martingales.
  This proves that $X^{(a)}$ and $X^{(qc)}$ both satisfy NUPNR$(\mathbb{H})$.\\
 \noindent Conversely, if $X^{(a)}$ and $X^{(qc)}$ satisfy NUPNR$(\mathbb{H})$,
 then there exist two $\mathbb{H}$-predictable real-valued processes $\phi_1, \phi_2>0$ and
 two positive $\mathbb{H}$-local martingales $D_1={\cal E}(N_1), D_2={\cal E}(N_2)$  such that $D_1(\phi_1\centerdot (I_{\Gamma_X}\centerdot S))$ and
$D_2(\phi_2\centerdot (I_{{\Gamma_X}^c}\centerdot X))$ are both
$\mathbb{H}$-local martingales. Remark that there is no loss of
generality in assuming $N_1=I_{\Gamma_X}\is N_1$ and
$N_2=I_{{\Gamma_X}^c}\is N_2$. Put
 \begin{eqnarray*}
  N := I_{{\Gamma_X}}\centerdot N_1  + I_{{\Gamma_X}^c}\centerdot N_2 \ \ \ \ \mbox{and} \ \ \ \psi := \phi_1 I_{{\Gamma_X}} + \phi_2I_{{\Gamma_X}^c}.
 \end{eqnarray*}
 Obviously, ${\cal E}( N)>0$, ${\cal E}( N)$ and ${\cal E}( N)(\psi \centerdot S)$ are $\mathbb H$-local martingales, $\psi$ is $\mathbb H$-predictable and $0< \psi\leq 1$. This ends the proof of the lemma.
\end{proof}

\noindent Therefore, throughout the paper $S$ is assumed to be a thin $\mathbb F$-semimartingale.
\noindent This paper is organized as follows. The next section (Section \ref{SectionBeforeTau}) addresses the case of stopping at $\tau$ (i.e. deals with the model $(S^{\tau},\mathbb G)$), while Section \ref{SectionAfterTau} focuses on the model $(S-S^{\tau},\mathbb G)$. Sections \ref{proofs} and \ref{Proof4theoremsNotTech} prove the main results elaborated in Sections \ref{SectionBeforeTau} and \ref{SectionAfterTau}. Section \ref{Proof4theoremsNotTech} is the most technical part of the paper. We conclude this paper with an appendix, where we recall some useful technical results.
%%%%%%%%%%%%%%%%%%%%%%%%%%%%%%%%%%%%%%%%%%%%%%%%%%%%%%%%%%%%%%%%%%%%%%%%%%%%%%%%%%%%%%%%%%%%%%%%%%%%%%%%%%%%%%%%%%%%%%
%%%%%%%%%%%%%%%%%%%%%%%%%%%%%%%%%%%%%%%%%%%%%%%%%%%%%%%%%%%%%%%%%%%%%%%%%%%%%%%%%%%%%%%%%%%%%%%%%%%%%%%%%%%%%%%%%%%%%%%%%%%%%%%%%%%%%%%%%%%%%%%%%
%%%%%%%%%%%%%%%%%%%%%%%%%%%%%%%%%%%%%%%%%%%%%%%%%%%%%%%%%%%%%%%%%%%%%%%%%%%%%%%%%%%%%%%%%%%%%%%%%%%%%%%%%%%%%%%%%%%%%%%%%%%%%%%%%%%%%%%%%%%%%%%%%
%%%%%%%%%%%%%%%%%%%%%%%%%%%%%%%%%%%%%%%%%%%%%%%%%%%%%%%%%%%%%%%%%%%%%%%%%%%%%%%%%%%%%%%%%%%%%%%%%%%%%%%%%%%%%%%%%%%%%%%%%%%%%%%%%%%%%%%%%%%%%%%%%
%%%%%%%%%%%%%%%%%%%%%%%%%%%%%%%%%%%%%%%%%%%%%%%%%%%%%%%%%%%%%%%%%%%%%%%%%%%%%%%%%%%%%%%%%%%%%%%%%%%%%%%%%%%%%%%%%%%%%%%%%%%%%%%%%%%%%%%%%%%%%%%%%%%%%%%%
%
%
%%%%%%%%%%%%%%%%%%%%%%%%%%%%%%%%%%%%%%%%%%%%%%%%%%%%%%%%%%%%%%%%%%%%%%%%%%%%%%%%%%%%%%%%%%%%%%%%%%%%%%%%%%%%%%%%%%%%%%%%%%%%%%%%%%%%%%%%%%%%%%%%%%%%%%%%%

%%%%%%%%%%%%%%%%%%%%%%%%%%%%%%%%%%%%%%%%%%%%%%%%%%%%%%%%%%%%%%%%%%%%%%%%%%%%%
%%%%%%%%%%%%%%%%%%%%%%%%%%
%%%%%%%%%%%%%%%%%%%%%%%%%%%%%%%%%%%%%%%%%%%%%%%%%%%%%%%%%%%%%%%%%%%%%%%%%%%%%%%
\section{The Case of Stopping at $\tau$ }\label{SectionBeforeTau}
This section elaborates our results on the NUPBR for the model $\left(S^{\tau},\mathbb G\right)$ in two subsections. The first subsection presents our principal results as well as their immediate consequences and/or applications, while the second subsection outlines a method to construct explicitly $\mathbb G$-local martingale deflators from $\mathbb F$-local martingale deflators. To this end, in addition to $\mathbb G$ and $D$ defined in (\ref{DandfiltrationG}), we associate to $\tau$ two important
$\mathbb F$-supermartingales given by
\begin{equation}\label{ZandZtilde}
Z_t := P\Bigl(\tau >t\ \big|\ {\cal F}_t\Bigr)\   \mbox{and}\ \ \ \widetilde Z_t:=P\left(\tau\geq t\ \Big|\ {\cal F}_t\right).
\end{equation}
The supermartingale $Z$ is right-continuous with left limits and
coincides with the $\mathbb F$-optional projection of $I_{\Lbrack
0, \tau\Rbrack}$, while $\widetilde Z$ admits right limits and
left limits only and is the $\mathbb F$-optional projection of
$I_{\Lbrack 0, \tau\Lbrack}$.  The decomposition of $Z$ leads to
an important $\mathbb F$-martingale  $m$,   given by
 \begin{equation}\label{processmZ}
m := Z+D^{o,\mathbb F},\end{equation}
 where $D^{o,\mathbb F}$ is the $\mathbb F$-dual optional projection of $D$
 (see  \cite{Jeu} for more details).

%%%%%%%%%%%%%%%%%%%%%%%%%%%%%%%%%%%%%%%%%%%%%%%%%%%%%%%%%%%%%%%%%%%%%%%%%%%%%%%%%%%%%%%%%%%%%%%%%%%%%%%%%%%%
%%%%%%%%%%%%%%%%%%%%%%%%%%%%%%%%%%%%%%%%%%%%%%%%%%%%%%%%%%%%%%%%%%%%%%%%%%%%%%%%%%%%%%%%%%%%%%%%%%%%%%%%%%%
%%%%%%%%%%%%%%%%%%%%%%%%%%%%%%%%%%%%%%%%%%%%%%%%%%%%%%%%%%%%%%%%%%%%%%%%%%%%%%%%%%%%%%%%%%%%%%%%%%%%%%%%%%%
%%%%%%%%%%%%%%%%%%%%%%%%%%%%%%%%%%%%%%%%%%%%%%%%%%%%%%%%%%%%%%%%%%%%%%%%%%%%%%%%%%%%%%%%%%%%%%%%%%%%%%%%%

\subsection{The main results}

In this subsection, we  outline the main results on the NUPBR
condition for the stopped thin $\mathbb
F$-semimartingales (with predictable jumps only) with $\tau$. To this end, we start by addressing the case of single jump processes with $\mathbb F$-predictable stopping times.

\begin{theorem}\label{main3} Consider an $\mathbb F$-predictable stopping time $T$
and an ${\cal F}_T$-measurable variable $\xi$ satisfying $E(\vert \xi\vert\big|{\cal F}_{T-})<+\infty$ P-a.s. on $\{T<+\infty\}$.\\ If $S:= \xi  I_{\{ Z_{T-}>0\}}I_{\Rbrack T,+\infty\Rbrack}$, then the following assertions are equivalent:\\
{\rm{(a)}} $S^{\tau}$ satisfies NUPBR$(\mathbb G)$.\\
{\rm{(b)}} The process $\widetilde S:=\xi I_{\{ \widetilde Z_T>0\}}I_{\Rbrack T,+\infty\Rbrack}=
I_{\{ \widetilde Z>0\}}\centerdot S$ satisfies NUPBR$(\mathbb F)$.\\
{\rm{(c)}} $S$ satisfies NUPBR$(\mathbb F,{\widetilde Q}_T)$,
where ${\widetilde Q}_T$ is
\begin{equation}\label{Qtilde(T)}
{\widetilde Q}_T:=\left({{\widetilde Z_T}\over{Z_{T-}}}I_{\{ Z_{T-}>0\}}+I_{\{ Z_{T-}=0\}}\right)\cdot P,\end{equation}
{\rm{(d)}} $S$ satisfies NUPBR$(\mathbb F,Q_T)$, where $Q_T$ is defined by
\begin{equation}\label{QT}
 {{dQ_T}\over{dP}}:=
 {{I_{\{ \widetilde Z_T>0\}\cap\Gamma_0(T)}}\over{P(\widetilde Z_T>0\big|\ {\cal F}_{T-})}}
 +I_{\Omega\setminus\Gamma_0(T)},\ \Gamma_0(T):=\{P(\widetilde Z_T>0|{\cal F}_{T-})>0\}. \end{equation}
\end{theorem}

\noindent The proof of this theorem is long and requires a result from the next subsection. Thus, this proof is delegated to Section  \ref{proofs}.\\

\begin{remark} \label{remark214}(a) The importance of Theorem \ref{main3} goes beyond its vital role,
as a building block for  the more general result. In fact,  Theorem \ref{main3} provides two different characterizations for   the NUPBR$\left(\mathbb G\right)$ of $S^{\tau}$. The characterizations (c) and (d) are expressed in term of the  NUPBR$\left(\mathbb F\right)$ of $S$ under absolute continuous change of measure, while the characterization (a) uses transformation of $S$ without any change of measure. Furthermore, Theorem \ref{main3} can be easily extended to the case of countably many ordered predictable jump times $T_0=0\leq T_1\leq T_2\leq...$ with $\sup_n T_n=+\infty\ P-a.s.$.\\
(b) In Theorem \ref{main3}, the choice of $S$ having the form $S:=
\xi I_{\{ Z_{T-}>0\}}I_{\Rbrack T,+\infty\Rbrack}$ is not
restrictive. This can be understood from the fact that any single
jump process $S$ can be decomposed as follows
$$S:= \xi  I_{\Rbrack T,+\infty\Rbrack}= \xi  I_{\{ Z_{T-}>0\}}I_{\Rbrack T,+\infty\Rbrack}+ \xi  I_{\{ Z_{T-}=0\}}I_{\Rbrack T,+\infty\Rbrack}
=:{\overline S}+{\widehat S}.$$ Thanks to $\{T\leq\tau\}\subset\{
Z_{T-}>0\}$, we have ${\widehat S}^\tau=\xi  I_{\{
Z_{T-}=0\}}I_{\{T\leq\tau\}}I_{\Rbrack T,+\infty\Rbrack}\equiv 0$
is (obviously) a $\mathbb G$-martingale. Thus, the only part of
$S$ that requires careful attention is ${\overline S}:=\xi I_{\{
Z_{T-}>0\}}I_{\Rbrack T,+\infty\Rbrack}$.
\end{remark}

\noindent The following proposition describes the models of $\tau$ for which any single jump $\mathbb F$-martingale (that jumps at fixed $\mathbb F$-predictable stopping time $T$), stopped at $\tau$, satisfies the NUPBR$(\mathbb G)$.

\begin{proposition}\label{corollaryofmain3}
Let $T$ be an $\mathbb F$-predictable stopping time. Then, the following assertions are equivalent:\\
{\rm{(a)}} On $\left\{ T<+\infty\right\}$, we have
\begin{equation}\label{equation1111}
 \left\{ \widetilde Z_T=0\right\}\subset\displaystyle\Bigl\{ Z_{T-}=0\Bigr\}.\end{equation}
{\rm{(b)}} For any $M:=\xi I_{\Rbrack T,+\infty\Rbrack}$ where $\xi\in L^{\infty}({\cal F}_T)$ such that $E(\xi|{\cal F}_{T-})=0$, $M^{\tau}$ satisfies NUPBR$(\mathbb G)$.
\end{proposition}

\begin{proof} We start by  proving ${\rm{(a)}} \Rightarrow {\rm{(b)}} $.
Suppose that (\ref{equation1111}) holds. Then, due to Remark
\ref{remark214}--(b), we can restrict our attention to the case where
  $M:=\xi I_{\{ Z_{T-}>0\}}I_{\Rbrack T,+\infty\Rbrack}$ with $\xi\in L^{\infty}({\cal F}_T)$ and $E(\xi|{\cal F}_{T-})=0$. Since assertion (a) is equivalent to $\Rbrack T\Lbrack\cap\{\widetilde Z=0\ \&\ Z_{-}>0\}=\emptyset$, we deduce that
  $$\widetilde M:=\xi I_{\{ {\widetilde Z}_{T}>0\}}I_{\{ Z_{T-}>0\}}I_{\Rbrack T,+\infty\Rbrack}=M\ \ \ \ \ \mbox{is an $\mathbb F$-martingale}.$$
  Therefore, a direct application of Theorem \ref{main3} (to $M$)  allows us to conclude that $M^{\tau}$ satisfies the NUPBR$(\mathbb G)$.
  This ends the proof of (a)$\Rightarrow $ (b). To prove the reverse implication, we suppose that assertion (b) holds and consider
$$
M:=\xi I_{\Rbrack T,+\infty\Rbrack},\ \ \ \ \mbox{where}\ \ \xi:=\left(I_{\{ \widetilde Z_T=0\}}-P( \widetilde Z_T=0|{\cal F}_{T-})\right)I_{\{T<+\infty\}}.$$
 Since $\{T\leq \tau\}\subset\{\widetilde Z_T>0\}\subset\{Z_{T-}>0\}$, then we get
$$
M^{\tau}=-P( \widetilde Z_T=0|{\cal F}_{T-})I_{\{ T\leq\tau\}}I_{\Rbrack T,+\infty\Rbrack},$$
and this process is $\mathbb G$-predictable. Therefore, $M^{\tau}$ satisfies NUPBR$(\mathbb G)$ if and only if it is a constant process equal to $M_0=0$ (see Lemma \ref{NUPBRforPredictableProcesses}). This is equivalent to
$$
0=E\Bigl[P( \widetilde Z_T=0|{\cal F}_{T-})I_{\{ T\leq\tau\}}I_{\Rbrack T,+\infty\Rbrack}\Bigr]=E\left(Z_{T-}I_{\{ \widetilde Z_T=0\ \&\ T<+\infty\}}\right).$$
It is obvious that this equality is equivalent to (\ref{equation1111}), and assertion (a) follows. This ends the proof of the theorem.
\end{proof}

\noindent The next theorem is an extension of Theorem
\ref{main3} to  the case where there are countable many arbitrary
predictable jumps, and constitutes our first main result for the general thin semimartingales with predictable jumps only.

\begin{theorem}\label{main4}
Let $S$ be a thin process with predictable jump times only.
Then, the following assertions are equivalent.\\
{\rm{(a)}}  The process $S^{\tau}$ satisfies the NUPBR$(\mathbb G)$.\\
{\rm{(b)}}  For any $\delta>0$, there exists a positive $\mathbb
F$-local martingale, $Y$, such that $\ ^{p,\mathbb F}\left(Y\vert \Delta S\vert I_{\{ \widetilde Z>0\}}\right)<+\infty$ P-a.s. on $\{Z_{-}\geq \delta\}$ and
\begin{equation}\label{mainassumptionbeforetau}
 ^{p,\mathbb F}\left(Y{\Delta S} I_{\{ \widetilde Z>0\}}\right)I_{\{ Z_{-}\geq\delta\}}=0.\end{equation}
 {\rm{(c)}} For any $\delta$, the process
\begin{equation}\label{Szero}
S^{(0)}:=\sum \Delta S I_{\{\widetilde Z>0\ \&\ Z_{-}\geq\delta\}}=I_{\{Z_{-}\geq\delta\}}\cdot S-\sum \Delta S I_{\{ \widetilde Z=0\ \&\ Z_{-}\geq\delta\}},\end{equation}
satisfies the NUPBR$(\mathbb F)$.
\end{theorem}

%%%%%%%%%%%%%%%%%%%%%%%%%%%%%%%%%%%%%%%%%%%%%%%%%%%%%%%%%%%%%%%%%%%%%%%%%%%%%%%%%%%%%%%%%%%%
\noindent  The proof of this theorem is technically involved, especially the proof of (a)$\Longrightarrow$(c), and thus it is postponed to
Subsection \ref{proofmain3}.

\begin{remark}\label{remarkformmain4} It is important to notice that, in Theorem \ref{main4}, we did not assume any arbitrage condition on $S$. Therefore, as consequence, we obtain the following. Suppose that $S$ is a thin process --with predictable jumps
only-- satisfying NUPBR$(\mathbb F)$ and
$$\{\widetilde Z=0\ \&\ Z_{-}>0\}\cap\{\Delta S\not=0\}=\emptyset.$$
  Then, S$^{\tau}$ satisfies NUPBR$(\mathbb G)$. This follows immediately from Theorem \ref{main4} by using $Y\in {\cal L}(S,\mathbb F)$ and Lemma \ref{LY}.
\end{remark}

%%%%%%%%%%%%%%%%%%%%%%%%%%%%%%%%%%%%%%%%%%%%%%%%%%%%%%%%%%%%%%%%%%%%%%%%%%%%%%%%%%%%%%%%%%%
%%%%%%%%%%%%%%%%%%%%%%%%%%%%%%%%%%%%%%%%%%%%%%%%%%%%%%%%%%%%%%%%%%%%%%%%%%%%%%%%%%%%%%%%%%%%%%%%%%%
\noindent The following extends Proposition \ref{corollaryofmain3} to the case of countably many jumps that might not be ordered in any way.

\begin{theorem}\label{XtauArbitrary}
The following assertions are equivalent.\\
{\rm{(a)}} The set $\{{\widetilde Z}=0>Z_{-}\}$ is totally inaccessible.\\
{\rm{(b)}} $X^{\tau}$ satisfies the NUPBR$(\mathbb G)$ for any thin process $X$ with predictable jumps satisfying NUPBR$(\mathbb F)$.
\end{theorem}

\begin{proof} The proof of the theorem will be achieved in two parts, namely part 1) and part 2) where we prove (b)$\Longrightarrow$(a) and (a)$\Longrightarrow$(b) respectively.\\
{\bf 1)} Suppose that assertion (b) holds. Then, thanks to Proposition \ref{corollaryofmain3}, we deduce that for any  $\mathbb F$-predictable stopping time $T$,
\begin{equation}\label{Equa400}
\Rbrack T\Lbrack\cap \{\widetilde Z=0<Z_{-}\}=\emptyset\end{equation}
on the one hand. On the other hand, since $\{\widetilde Z=0<Z_{-}\}$ is thin, there exists a sequence of $\mathbb F$-stopping times $(\sigma_k)_{k\geq 1}$ with disjoint graphs such that
\begin{equation}\label{equa410}
\{\widetilde Z=0<Z_{-}\}=\bigcup_{k=1}^{+\infty}\Rbrack \sigma_k\Lbrack.\end{equation}
Recall that, for each $\sigma_k$, there exist two $\mathbb F$-stopping times ($\sigma_k^i$ and $\sigma_k^a$ that are totally inaccessible and accessible respectively) and a sequence of $\mathbb F$-predictable stopping times $(T_l^{(k)})_{l\geq 1}$ such that
$$\Rbrack \sigma_k\Lbrack=\Rbrack \sigma_k^i\Lbrack\cup \Rbrack \sigma_k^a\Lbrack,\ \ \ \ \Rbrack \sigma_k^a\Lbrack\subset \bigcup_{l=1}^{+\infty}\Rbrack T_l^{(k)}\Lbrack.$$
Thus, by combining these with $\displaystyle\left(\bigcup_{k=1}^{+\infty} \Rbrack \sigma_k^i\Lbrack\right)\cap\left(\bigcup_{k=1,l=1}^{+\infty}\Rbrack T_l^{(k)}\Lbrack\right)=\emptyset$, (\ref{equa410}) and (\ref{Equa400}), we derive
$$\bigcup_{k=1}^{+\infty} \Rbrack \sigma_k^a\Lbrack=\left(\bigcup_{k=1,l=1}^{+\infty}\Rbrack T_l^{(k)}\Lbrack\right)\cap \{\widetilde Z=0<Z_{-}\}=\emptyset.$$
This proves that $\{\widetilde Z=0<Z_{-}\}$ is a totally inaccessible set and the proof of (b)$\Longrightarrow$(a) is completed.\\
{\bf 2)} To prove the reverse sense, we assume that assertion (a) holds, and consider $X=\sum \xi_n I_{\Rbrack T_n,+\infty\Rbrack}$ satisfying NUPBR$(\mathbb F)$, where $T_n$ is an $\mathbb F$-predictable stopping time and $\xi_n$ is a bounded ${\cal F}_{T_n}$-measurable random variable. Since $\{\Delta X\not=0\}=\displaystyle\bigcup_{n=1}^{+\infty}\Rbrack T_n\Lbrack$ is predictable, we get $\{\widetilde Z=0<Z_{-}\}\cap\{\Delta X\not=0\}=\emptyset$, and hence, from Remark \ref{remarkformmain4}, $X^{\tau}$ satisfies the NUPBR$(\mathbb G)$. This ends the proof of the theorem.\end{proof}

\noindent The complete general result, in this spirit of describing the model for $\tau$ that preserves the NUPBR after stopping with $\tau$, is the following.

\begin{theorem}\label{XtauArbitrary1}
The following assertions are equivalent.\\
{\rm{(a)}} The set $\{{\widetilde Z}=0>Z_{-}\}$ is evanescent.\\
{\rm{(b)}} $X^{\tau}$ satisfies the NUPBR$(\mathbb G)$ for any $X$ satisfying the NUPBR$(\mathbb F)$.
\end{theorem}

\begin{proof} The proof follows immediately from the combination of Theorem \ref{XtauArbitrary} and Proposition 2.22 in \cite{aksamit/choulli/deng/jeanblanc2} (where the authors prove that the thin set $\{\widetilde Z=0<Z_{-}\}$ is accessible if and only if assertion (b) above holds for any $\mathbb F$-quasi-left-continuous process $X$) .
\end{proof}
%%%%%%%%%%%%%%%%\subsection{The Case of General Semimartingales}

\subsection{Explicit local martingale deflators}\label{proofsOfmainTheorems3}
This section discusses how to construct explicitly $\mathbb G$-local martingale deflators from $\mathbb F$-deflators for a class of processes. This is achieved, for single jump processes and general thin processes afterwards,  by considering $\mathbb F$-neutralized processes.

 \begin{proposition}\label{cruciallemma1} Let $M:=\xi I_{\Rbrack T,+\infty\Rbrack}$ be an $\mathbb F$-martingale, where  $T$ is an $\mathbb F$-predictable stopping time, and $\xi$ is an ${\cal F}_T$-measurable random variable. Then the following assertions are equivalent.\\
{\rm{(a)}} $M$ is an $\mathbb F$-martingale under $Q_T$ given by (\ref{QT}).\\
{\rm{(b)}} On the set $\{ T<+\infty\}$, we have
  \begin{equation}\label{zeroequationbeforetau}
E\left(M_T I_{\{{\widetilde Z}_T=0\}}\big|\ {\cal
F}_{T-}\right)=0,\ \ \ P-a.s.\end{equation}
{\rm{(c)}} $M^{\tau}$
is a $\mathbb G$-martingale under $Q^{\mathbb
G}_T$ given by
\begin{equation}\label{QGbeforetau}
{{dQ^{\mathbb
G}_T}\over{dP}}:={{U^{\mathbb G} (T)}\over{E(U^{\mathbb G} (T)\big|\ {\cal
G}_{T-})}}\ \mbox{where}\ U^{\mathbb
G}(T):=I_{\{T>\tau\}}+I_{\{T\leq\tau\}}{{Z_{T-}}\over{\widetilde
Z_{T}}}. \end{equation}
\end{proposition}

\begin{proof} The proof will be achieved in two steps where we prove (a)$\Longleftrightarrow$(b) and (a)$\Longleftrightarrow$(c) respectively.\\
{\bf Step 1.} Here, we prove (a)$\Longleftrightarrow$(b). For simplicity we denote by $Q:=Q_T$, where $Q_T$ is
defined in (\ref{QT}), and remark that on $\{ Z_{T-}=0\}$, $Q$
coincides with $P$ and (\ref{zeroequationbeforetau}) holds, due to
$\{Z_{T-}=0\}\subset \{\widetilde Z_{T}=0\}$. Thus, it is enough
to prove (a)$\Longleftrightarrow$(\ref{zeroequationbeforetau}) on the set
$\{T<+\infty\ \&\ Z_{T-}>0\}$. On this set, due to $E(\xi |{\cal
F}_{T-})=0$ (since $M$ is an $\mathbb F$-martingale), we derive

 \begin{eqnarray*} \label{equa30}
 E^{Q}(\xi\big|{\cal F}_{T-})&=&E(\xi I_{\{ \widetilde Z_T>0\}}\big|{\cal F}_{T-})\Bigl(P(\widetilde Z_T>0\big|{\cal F}_{T-})\Bigr)^{-1}\\
 & =&-E(\xi I_{\{ \widetilde Z_T=0\}}\big|{\cal F}_{T-})\left(P(\widetilde Z_T>0\big|{\cal F}_{T-})\right)^{-1}.
\end{eqnarray*}
Therefore, assertion (a) (or equivalently $E^Q(\xi|{\cal F}_{T-})=0$) is equivalent to (\ref{zeroequationbeforetau}). This ends the proof of (a) $\Longleftrightarrow$ (b).\\
{\bf Step 2.} To prove (a)$\Longleftrightarrow$(c), we notice that due to $\{T\leq\tau\}\subset\{\widetilde Z_{T}>0\}\subset
\{Z_{T-}>0\}$, on $\{T\leq\tau\}$ we have

\begin{eqnarray}
  P\left(\widetilde Z_T>0\big|{\cal F}_{T-}\right)E^{Q^\mathbb{G}_T} \left(\xi | \mathcal{G}_{T-}\right) &=& E \left( \frac{Z_{T-}}{\widetilde{Z}_T} \xi I_{\{ T \leq \tau\}} | \mathcal{G}_{T-}\right)=E\left( \xi I_{\{ \widetilde{Z}_T > 0\}} | \mathcal{F}_{T-}\right)   \nonumber \\
  &=&  E^Q\left(\xi | \mathcal{F}_{T-}\right)  P\left(\widetilde Z_T>0\big|{\cal F}_{T-}\right) \nonumber.
\end{eqnarray}

%%%%%%%%%%%%%%%%%%%%%%%%%%%%%%%%%%%%%%%%%%%%%%%%%%%%%%%%%%%%%%%%%%%%%%%%%%%%%%%%%%%%%%%%%%%%%%%%%%%%%%%%%%%%%%%%

\noindent This equality proves that $M^{\tau}\in {\cal M}(Q^{\mathbb G},\mathbb G)$
  if and only if $M\in {\cal M}(Q,\mathbb F)$, and the proof of (a)$\Longleftrightarrow$(c) is completed. This ends the proof of the theorem.\end{proof}

To generalize this proposition to the case of infinitely many jumps that might not be ordered at all, we need to introduce some notations and recall some facts from \cite{aksamit/choulli/deng/jeanblanc2}. First of all, we refer to \cite{dm2} ( Chapter VIII.2 sections 32-35 pages 356-361) and \cite{Jacod} ( Chapter III.4.b, Definition 3(3.8), pages 106-109) for the optional stochastic integration (see also Definition 3.4 in \cite{aksamit/choulli/deng/jeanblanc2}).
\begin{definition}\label{OIDefinition}
Let $N$ be an $\mathbb H$-local martingale with continuous part $N^c$ and $K$ be an ${\mathbb H}$-optional process. $K$ is said to be integrable with respect to $N$ if $^{p,\mathbb H}(K)$ is $N^c$-integrable, $^{p,\mathbb H}(K\vert \Delta N\vert)<+\infty$ and $$\left(\sum (K\Delta N-\ ^{p,\mathbb H}(K\Delta N)\right)^{1/2}\in {\cal A}^+_{loc}(\mathbb H).$$
\end{definition}
Put
\begin{equation}\label{G-martingaleL(a)}
K^{\mathbb G}:={{Z_{-}^2{\widetilde Z}^{-1}}\over{Z_{-}^2+\Delta\langle m\rangle^{\mathbb F}}}I_{\Lbrack 0,\tau\Lbrack},\ \ \ V^{\mathbb G}:=\sum\ ^{p,\mathbb F}(I_{\{\widetilde Z=0\}})I_{\Lbrack 0,\tau\Lbrack},\end{equation}
and to any $\mathbb F$-local martingale $M$, we associate the $\mathbb G$-local martingale part of $M^{\tau}$ given by
\begin{equation}\label{MG-martingaleprt}
\widehat M:=M^{\tau}-Z_{-}^{-1}I_{\Rbrack 0,\tau\Lbrack}\cdot \langle M,m\rangle^{\mathbb F}.\end{equation}

\noindent Below, we recall some useful results of \cite{aksamit/choulli/deng/jeanblanc2}.

\begin{proposition}\label{ResultoofACDJ1} The following assertions hold.\\
(a) The $\mathbb G$-optional process $K^{\mathbb G}$ is $\widehat m$-integrable in the sense of the above definition. Here $\widehat m:=m^{\tau}-Z_{-}^{-1}I_{\Lbrack 0,\tau\Lbrack}\cdot \langle m\rangle^{\mathbb F}$. Furthermore the resulting integral \\
\begin{equation}\label{GmgLtilde(a)}
{\widetilde L}^{(b)}:={\cal E}\left(-{{K^{\mathbb G}}\over{1-\Delta V^{\mathbb G}}}\odot \widehat m\right),\end{equation}
is a positive (i.e. ${\widetilde L}^{(b)}>0$) $\mathbb G$-local martingale satisfying $[{\widetilde L}^{(b)}, M]\in {\cal A}_{loc}(\mathbb G)$ for any $\mathbb F$-local martingale $M$.\\
(b) $V^{\mathbb G}\in {\cal A}^+_{loc}(\mathbb G)$ and $(1-\Delta V^{\mathbb G})^{-1}$ is $\mathbb G$-locally bounded.
\end{proposition}

\noindent The proof of this proposition can be found in \cite{aksamit/choulli/deng/jeanblanc2} (see Lemma 3.3 and Proposition 3.6). The extension of Proposition \ref{cruciallemma1} goes through connecting the random variable $U^{\mathbb G}(T)$ defined in (\ref{QGbeforetau}) to the process ${\widetilde L}^{(a)}$ as follows.

\begin{remark}\label{remarkfor Ltilde(b)} In virtue of the calculation performed in \cite{aksamit/choulli/deng/jeanblanc2} (see equation (B.1) where the authors calculate the jumps of $K^{\mathbb G}\odot \widehat m$), we have
$$
-{{(1-\Delta V^{\mathbb G})\Delta {\widetilde L}^{(b)}}\over{{\widetilde L}^{(b)}_{-}}}=K^{\mathbb G}\Delta{\widehat m}-\ ^{p,\mathbb G}\left(K^{\mathbb G}\Delta{\widehat m}\right)={{\Delta m}\over{\widetilde Z}}I_{\Lbrack 0,\tau\Rbrack}-\Delta V^{\mathbb G}.$$
Thus, for an $\mathbb F$-predictable stopping time $T$, on $\{T\leq\tau\}$ we get
$$U^{\mathbb G}_T={{Z_{T-}}\over{\widetilde Z_T}}=(1-\Delta V^{\mathbb G}_T){{{\widetilde L}^{(b)}_T}\over{{\widetilde L}^{(b)}_{T-}}}.$$
This proves that assertions (a) and (b) of Proposition \ref{cruciallemma1} are equivalent to
\begin{equation}\label{conclusion1}
{\widetilde L}^{(b)}M^{\tau}\ \mbox{is a}\ {\mathbb G}\mbox{-martingale for any single jump $\mathbb F$-martingale}\ M.\end{equation}
\end{remark}

\begin{theorem}\label{generalThinDeflators}
Consider ${\widetilde L}^{(b)}$ defined in (\ref{GmgLtilde(a)}) and let $M$ be a thin $\mathbb F$-martingale satisfying
\begin{equation}\label{Condition1}
^{p,\mathbb F}\left(\Delta M I_{\{\widetilde Z=0<Z_{-}\}}\right)\equiv 0.\end{equation}
Then, ${\widetilde L}^{(b)}M^{\tau}$ is a $\mathbb G$-local martingale.
\end{theorem}

\begin{proof} We start by remarking that it is enough to prove that there exists a $\mathbb G$-predictable process $\varphi$ such that $0<\varphi\leq 1$ and ${\widetilde L}^{(b)}(\varphi\cdot M^{\tau})$ is a $\mathbb G$-martingale (local martingale). This means that ${\widetilde L}^{(b)}\in {\cal L}(M^{\tau},\mathbb G)$ (i.e it is a $\sigma$-martingale density for $M^{\tau}$ under $\mathbb G$). This remark that simplifies the proof based on the fact that $[{\widetilde L}^{(b)}, M^{\tau}]$ is locally integrable and Proposition 3.3 and Corollary 3.5 of \cite{anselstricker1994}. Again, thanks to $[{\widetilde L}^{(b)}, M^{\tau}]\in{\cal A}_{loc}(\mathbb G)$, we deduce that $^{p,\mathbb G}\left({\widetilde L}^{(b)}\vert \Delta M^{\tau}\vert\right)<+\infty$, and consider the following $\mathbb G$-predictable process
$$
\phi:=\left[1+^{p,\mathbb G}\left(\vert \Delta M^{\tau}\vert\right)+\ ^{p,\mathbb G}\left({\widetilde L}^{(b)}\vert \Delta M^{\tau}\vert\right)\right]^{-1}\left[I_{\Omega\setminus(\cup_n\Rbrack T_n\Lbrack)}+\sum_{n=1}^{+\infty} 2^{-n} I_{ \Rbrack T_n\Lbrack}\right],$$ where $(T_n)_{n\geq 1}$ is the sequence of $\mathbb F$-predictable stopping times that exhausts the jumps of $M$. Thus, it is easy to check that $0<\phi\leq 1$, and both processes $\phi\cdot M^{\tau}$ and ${\widetilde L}^{(b)}_{-}\phi\cdot M^{\tau}+[{\widetilde L}^{(b)},\phi\cdot M^{\tau}]=\sum {\widetilde L}^{(b)}\phi\Delta M^{\tau}$ have integrable variations on the one hand. On the other hand, since $ \sum {\widetilde L}^{(b)}\phi\Delta M^{\tau}$ jumps on predictable stopping times only, its $\mathbb G$-compensator is
$$\sum\ ^{p,\mathbb G}\left( {\widetilde L}^{(b)}\phi\Delta M^{\tau}\right)=\sum \phi\ ^{p,\mathbb G}\left( {\widetilde L}^{(b)}\Delta M^{\tau}\right)\equiv 0.$$
This proves that ${\widetilde L}^{(b)}_{-}\phi\cdot M^{\tau}+[{\widetilde L}^{(b)},\phi\cdot M^{\tau}]$ is a $\mathbb G$-local martingale or equivalently  ${\widetilde L}^{(b)}(\phi\cdot M^{\tau})$  is a $\mathbb G$-local martingale. This ends the proof of the theorem.\end{proof}

\begin{corollary} For any thin $\mathbb F$-martingale $M$ such that $\{\Delta M\not=0\}\cap\{\widetilde Z=0<Z_{-}\}$ is evanescent, ${\widetilde L}^{(b)} M^{\tau}$ is a $\mathbb G$-local martingale.
\end{corollary}

\begin{proof} The proof of the corollary follows immediately from Theorem \ref{generalThinDeflators}, as the condition $\{\Delta M\not=0\}\cap\{\widetilde Z=0<Z_{-}\}=\emptyset$ implies (\ref{Condition1}).
\end{proof}
%%%%%%%%%%%%%%%%%%%%%%%%%%%%%%%%%%%%%%%%%%%%%%%%%%%%%%%%%%%%%%%%%%%%%%%%%%%%%%%%%%%%%%%%%%%%%%%%%%%%%%%%%
%%%%%%%%%%%%%%%%%%%%%%%%%%%%%%%%%%%%%%%%%%%%%%%%%%%%%%%%%%%%%% OTHER RESULTS%%%%%%%%%%%%%%%%%%%%%%%%
\section{The part after $\tau$}\label{SectionAfterTau}
Herein, we focus on the process $S-S^{\tau}$, and in the same spirit of Section \ref{SectionBeforeTau} we summarize results in two subsections. The first subsection outlines the principal results, while the second subsection explains how to obtain $\mathbb G$-local martingale deflators for $S-S^{\tau}$ from the $\mathbb F$-deflators of $S$ when $S$ varies in a class of processes. However in this section we consider the following assumption on $\tau$
\begin{equation}\label{mainassumptionontau}
\tau\ \mbox{is an honest time and }\ \ \ \ \ \ Z_{\tau}<1\ \ \ P-a.s.\end{equation}

\subsection{The main results}
\noindent This subsection presents our main results on the NUPBR for $(S-S^{\tau},\mathbb G)$. These results are elaborated for single jump processes and general thin processes with predictable jumps only as well.

 \begin{theorem}\label{cruciallemma2} Suppose that $\tau$ is an honest time. Consider an $\mathbb F$-predictable stopping time $T$ and an ${\cal F}_T$-measurable
 r.v. $\xi$ such that
$E(\vert \xi\vert\big|\ {\cal F}_{T-})<+\infty$  P-a.s. on $\{T<+\infty\}$.\\ If $S:=\xi I_{\{Z_{T-}<1\}}I_{\Rbrack T,+\infty\Rbrack}$, then the following are equivalent:\\
{\rm{(a)}}  $S-S^{\tau}$ satisfies the NUPBR$(\mathbb G)$.\\
{\rm{(b)}}  $S$ satisfies the NUPBR$(\mathbb F,\widetilde{Q}'_T)$,
where
\begin{equation}\label{probabilityQ'tilde(T)}
{\widetilde Q}'_T:=\Big({{1-\widetilde Z_T}\over{1-Z_{T-}}}I_{\{
Z_{T-}<1\}}+I_{\{ Z_{T-}=1\}}\Bigr)\cdot P.\end{equation}
{\rm{(c)}}  $S$ satisfies the NUPBR$(\mathbb F, Q'_T)$, where for
$\Gamma_1(T):=\{P(\widetilde Z_T<1\big|{\cal F}_{T-})>0\ \&\ T<+\infty\}$ we set
\begin{equation}\label{probabilityQprime}
Q'_T:=\Bigl({{I_{\{\widetilde
Z_T<1\}\cap\Gamma_1(T)}}\over{P(\widetilde Z_T<1\big|{\cal
F}_{T-})}}+I_{\Omega\setminus\Gamma_1(T)}\Bigr)\cdot P.\end{equation} (d)
$\widetilde S:=\xi I_{\{ \widetilde Z_T<1\}}I_{\Rbrack
T,+\infty\Rbrack}$ satisfies the NUPBR$(\mathbb F)$.
\end{theorem}

\noindent The proof of this theorem is long and requires intermediary results.
 Thus, we postpone the proof to  Subsection \ref{proofmain3}.

\begin{remark}
Theorem \ref{cruciallemma2} provides two equivalent (and
conceptually different) characterisations for the condition that
$S-S^{\tau}$ satisfies   NUPBR$(\mathbb G)$. One of these
characterisations uses  the  NUPBR$(\mathbb F)$  property under
$P$ for a transformation of $S$, while the other characterisation
is essentially based on the NUPBR$(\mathbb F)$ for $S$
under an absolutely continuous probability measure.
\end{remark}

\noindent The next theorem describes the models for $\tau$ that preserve the NUPBR$(\mathbb G)$ after $\tau$ for any single jump $\mathbb F$-martingale.

\begin{theorem}\label{cruciallemma3} Suppose that $\tau$ is an honest and consider an $\mathbb F$-predictable stopping time $T$. Then, the  following assertions are equivalent:\\
{\rm{(a)}}    On $\{ T<+\infty\}$, we have
\begin{equation}\label{zeroequationaftertau1} \left\{{\widetilde
Z}_T=1\right\}\subset \left\{ Z_{T-}=1\right\}.\end{equation}
{\rm{(b)}} For any $\xi \in L^{\infty}({\cal F}_T)$ such that
$E(\xi \big|\ {\cal F}_{T-})=0$ P-a.s on $\{T<+\infty\}$, the process $M-M^{\tau}$
satisfies NUPBR$(\mathbb G)$, where $M:=\xi I_{\Rbrack
T,+\infty\Rbrack}$.
\end{theorem}

\begin{proof}
Suppose that assertion (a) holds, and consider $\xi\in
L^{\infty}({\cal F}_T)$ such that $E(\xi\ |\ {\cal F}_{T-})=0,\ \
P-a.s.$ on $\{T<+\infty\}$. By decomposing $M$ into $$M=
I_{\{Z_{T-}<1\}}\xi I_{\Rbrack T,+\infty\Rbrack}+I_{\{Z_{T-}=1\}}\xi
I_{\Rbrack T,+\infty\Rbrack}:=M^{(1)}+M^{(2)},$$  and noting that
$M^{(2)}-(M^{(2)})^{\tau}=0$, we can restrict our attention to the case
where $M=M^{(1)}$ on the one hand.  On the other hand, since $\{Z_{T-}=1\}\subset\{\widetilde Z_T=1\}$ P-a.s. on $\{T<+\infty\}$, it is obvious that  (\ref{zeroequationaftertau1}) implies $\{\widetilde
Z_T<1\}=\{Z_{T-}<1\}$ on $\{T<+\infty\}$, and hence
$${\widetilde M}:=I_{\{ \widetilde Z_T<1\}} M=M\ \ \ \mbox{is an $\mathbb F$-martingale}.$$
Thus, assertion (b) follows from a direct application of Theorem
\ref{cruciallemma2} to $M$. This ends the proof of
(a)$\Rightarrow$ (b). To prove the converse, we assume that
assertion (b) holds, and we consider the ${\cal F}_T$-measurable
and bounded r.v. $\xi:=(I_{\{ \widetilde Z_T=1\}}-P(\widetilde
Z_T=1|{\cal F}_{T-}))I_{\{T<+\infty\}}$ and the bounded $\mathbb
F$-martingale $M:=\xi I_{\Rbrack T,+\infty\Rbrack}$. Then,  on the
one hand, $M-M^{\tau}$ satisfies NUPBR($\mathbb G)$. On the other
hand, due to $\{T>\tau\}\subset\{\widetilde Z_T<1\}$, the
finite variation process
$$
M-M^{\tau}=-P(\widetilde Z_T=1|{\cal F}_{T-})I_{\{
T>\tau\}}I_{\Rbrack T,+\infty\Rbrack}\ \ \mbox{is}\ \mathbb
G-\mbox{predictable}.$$ Thus, it is null, or equivalently
$\{Z_{T-}<1\}\subset\{\widetilde Z_T<1\}$ $P-a.s.$ on
$\{T<+\infty\}$. This proves assertion (a), and the proof of the
theorem is completed.\end{proof}

\noindent The following extends Theorem \ref{cruciallemma2} to the case of general thin processes.

\begin{theorem}\label{MainTheoremMultiJumps}
Suppose that $\tau$ satisfies (\ref{mainassumptionontau}), and $S$ is a thin process with predictable jumps only. Then, the following assertions are equivalent.\\
{\rm{(a)}}  The process $S-S^{\tau}$ satisfies the  NUPBR$(\mathbb G)$.\\
{\rm{(b)}} For any $\delta>0$, there exists a positive $\mathbb
F$-local martingale  $Y$, such that
\begin{equation}\label{mainEquationMultiJumps}
\ \ \ \ \ ^{p,\mathbb F}\left(Y\vert\Delta S\vert I_{\{\widetilde Z<1\}}\right)<+\infty\ \&\ \ ^{p,\mathbb F}\left(Y\Delta S I_{\{ \widetilde Z<1\}}\right)=0\ \mbox{on}\ \{1-Z_{-}\geq \delta\}.\end{equation}
{\rm{(c)}} For any $\delta$, the process
\begin{equation}\label{Sone}
S^{(1)}:=\sum \Delta S I_{\{\widetilde Z<1\ \&\ 1-Z_{-}\geq\delta\}},\end{equation}
satisfies the NUPBR$(\mathbb F)$.\end{theorem}

\noindent The proof of this theorem is long and is based on a result of the next subsection. Thus, this proof  is  postponed to Subsection
\ref{proofofMainTheoremMultiJumps}.

\begin{remark}
1) The process $S^{(1)}$ defined in (\ref{Sone}) is a thin semimartingale. In fact, we have $S^{(1)}=I_{\{1-Z_{-}\geq\delta\}}\cdot S-\sum \Delta S I_{\{ \widetilde Z=1\ \&\ 1-Z_{-}\geq\delta\}}$, and $$\sum I_{\{ \widetilde Z=1\ \&\ 1-Z_{-}\geq\delta\}}\leq \delta^{-2}\sum (\Delta m)^2\leq \delta^{-2} [m,m]\in{\cal A}^+_{loc}(\mathbb F).$$
2) The proof of (a)$\Longrightarrow$(b) is the very technical part in the proof of the theorem, while the rest is easy and is postponed to keep this section short.
\end{remark}

\begin{theorem}\label{XtauArbitraryafteraTau}
The following assertions are equivalent.\\
{\rm{(a)}} The set $\{{\widetilde Z}=1>Z_{-}\}$ is totally inaccessible.\\
{\rm{(b)}} $X-X^{\tau}$ satisfies the NUPBR$(\mathbb G)$ for any thin process $X$ with predictable jumps satisfying NUPBR$(\mathbb F)$.
\end{theorem}

\begin{proof}Suppose that assertion (a) holds, and consider a thin process with predictable jumps, $X$, satisfying NUPBR$(\mathbb F)$. Thus, $\{\Delta X\not=0\}$ is a thin accessible set, and hence $\{\widetilde Z=1>Z_{-}\}\cap\{\Delta X\not=0\}=\emptyset$. Therefore, we conclude that
$$
X^{(1)}:=\sum \Delta X I_{\{ \widetilde Z<1\ \&\ 1-Z_{-}\geq\delta\}}=I_{\{1-Z_{-}\geq \delta\}}\cdot X\ \mbox{satisfies NUPBR}(\mathbb F).$$
Then, a direct application of Theorem \ref{MainTheoremMultiJumps} leads to the NUPBR$(\mathbb G)$ of $X-X^{\tau}$. This proves (a)$\Longrightarrow$(b). To prove the reverse, we remark that the set $\{\widetilde Z=1>Z_{-}\}$ is thin, and we mimic exactly the part 1) of the proof of Theorem \ref{XtauArbitrary}. This ends the proof of theorem.
\end{proof}

\begin{theorem}\label{XtauArbitraryAfterTau1}
The following assertions are equivalent.\\
{\rm{(a)}} The set $\{{\widetilde Z}=1>Z_{-}\}$ is evanescent.\\
{\rm{(b)}} $X-X^{\tau}$ satisfies the NUPBR$(\mathbb G)$ for any $X$ satisfying NUPBR$(\mathbb F)$.
\end{theorem}

\begin{proof} The proof follows immediately from the combination of Theorem \ref{XtauArbitraryafteraTau} and Proposition 2.18 in \cite{aksamit/choulli/deng/jeanblanc3}(where the authors prove that the this set $\{\widetilde Z=1>Z_{-}\}$ is accessible if and only if assertion (b) of the theorem above holds for any quasi-left-continuous process $X$ (i.e. $X$ does not jump on predictable stopping times).
\end{proof}

%%%%%%%%%%%%%%%%%%%%%%%%%%%%%%%%%%%%%%%%%%%%%%%%%%%%%%%%%%%%%%%%%%%%%%%%%%%%%%%%%%%%%%%%%%%%%%%%%%%%%%%%
%%%%%%%%%%%%%%%%%%%%%%%%%%%%%%%%%%%%%%%%%%%%%%%%%%%%%%%%%%%%%%%%%%%%%%%%%%%%%%%%%%%%%%%%%%%%%%%%%%%%%%%%%%
\subsection{Explicit construction of local martingale deflators}\label{SubsectionExplicitDeflator2}
To construct $\mathbb G$-deflators for thin $\mathbb F$-local martingale, we start by illustrating this construction for single jump $\mathbb F$-martingales.

\begin{theorem}\label{cruciallemma2singlejump}  Let $\tau$ be
an honest time.
  Consider an $\mathbb F$-predictable stopping time $T$ and an
  ${\cal F}_T$-measurable r.v. $\xi$ such that
  $E[|\xi|| {\cal F}_{T-}] <+\infty, P$-a.s. Define  $M:= \xi I_{\{Z_{T-}<1\}} I_{\Rbrack T,+\infty\Rbrack}$,
  \begin{eqnarray}\label{probabilityQ'}
    \frac{dQ^{\mathbb F}_T}{dP} &:=& D ^{\mathbb F} := \frac{I_{\{\widetilde{Z}_T <1\ \&\ P(\widetilde{Z}_T<1 | {\cal F}_{T-}) > 0\}}}{ P(\widetilde{Z}_T<1 | {\cal F}_{T-})} + I_{\{P(\widetilde{Z}_T<1 | {\cal F}_{T-}) = 0\}},  \ \mbox{ and }\nonumber \\
    &&\frac{dQ^{\mathbb G}_T}{dP} := D ^{\mathbb G} :=\frac{1 - Z_{T-}}{(1-\widetilde{Z}_T) P(\widetilde{Z}_T<1 | {\cal F}_{T-})} I_{\{T>\tau\}} + I_{\{T\leq \tau\}}.
  \end{eqnarray}
  Then the following assertions are equivalent.\\
  {\rm (a)} $M$ is a $(Q^{\mathbb F}_T,\mathbb F)$-martingale.\\
  {\rm (b)} On $\{ Z_{T-}<1\}$, we have
  \begin{equation}\label{mgaftertau}
  E\left(\xi I_{\{ \widetilde Z_T<1\}}\ \big|\ {\cal F}_{T-}\right)=0,\ \ \ \ P-a.s.\end{equation}
  {\rm (c)} $(M-M^\tau)$ is a  $(Q^{\mathbb G}_T, \mathbb G)$-martingale.
\end{theorem}
\begin{proof} For the sake of simplicity, throughout the proof, we put  $Q_1:=Q^{\mathbb F}_T$ and $Q_2:=Q^{\mathbb G}_T$.
  The proof of the theorem will be given in two steps.\\
  {\bf 1)} Here, we prove (a)$\Longleftrightarrow$(b). Thanks to $\{\widetilde Z_T<1\}\subset\{Z_{T-}<1\}$ and $E[D ^{\mathbb F} | {\cal F}_{T-}]=1$ on $\{T<+\infty\}$, we derive
  \begin{eqnarray*}
   E^{Q_1}[\xi I_{\{Z_{T-}<1\}} | {\cal F}_{T-}] &=& E\left[D ^{\mathbb F} \xi I_{\{Z_{T-}<1\}} | {\cal F}_{T-}\right]
   = \frac{E\left[\xi I_{\{\widetilde{Z}_{T}<1\}} | {\cal F}_{T-}\right]}{P(\widetilde{Z}_T<1 | {\cal F}_{T-})}I_{\{Z_{T-}<1\}}. \end{eqnarray*}
   Therefore, (a)$\Longleftrightarrow$(b) follows from combining this equality and the fact that $M$ is a $(Q_1,\mathbb F)$-martingale if and only if
   $E^{Q_1}(M_T\ \big|\ {\cal F}_{T-})I_{\{ T<+\infty\}}=0$.\\
   {\bf 2)} Here, we prove (b)$\Longleftrightarrow$ (c). To this end, we first notice that
$M-M^{\tau}=\xi I_{\{Z_{T-}<1\ \&\ T>\tau\}} I_{\Rbrack\tau,+\infty\Rbrack}$ is a $(Q_2, \mathbb G)$-martingale if and only if
$E^{Q_2}[\xi I_{\{Z_{T-}<1\ \&\ T>\tau\}} | {\cal G}_{T-}]I_{\{ T<+\infty\}} = 0$. Then, using the fact that $E[D ^{\mathbb G} | {\cal G}_{T-}] =1$ on $\{T<+\infty\}$, we get
  \begin{eqnarray}\label{equalities22}
  &&  E^{Q_2}[\xi I_{\{Z_{T-}<1\ \&\ T>\tau\}} | {\cal G}_{T-}] = E\left[D ^{\mathbb G}\xi I_{\{Z_{T-}<1\}} I_{\{T>\tau\}} | {\cal G}_{T-}\right]\nonumber\\
    &&=   E\left[ \frac{\xi I_{\{T>\tau\}}}{1 - \widetilde{Z}_T}
     \Big| {\cal G}_{T-}\right] \frac{1-Z_{T-}}{P(\widetilde{Z}_T<1 | {\cal F}_{T-})} I_{\{Z_{T-}<1\ \&\ T>\tau\}}\nonumber\\
   & &=  \frac{E\left[\xi I_{\{\widetilde{Z}_{T}<1\}} | {\cal F}_{T-}\right]}{P(\widetilde{Z}_T<1 | {\cal F}_{T-})}I_{\{Z_{T-}<1\}}I_{\{T>\tau\}},
  \end{eqnarray}
  where the last equality in (\ref{equalities22}) follows from the fact that, $\tau$ being honest and
   $$
   E\left(H\ \big|\ {\cal G}_{T-}\right)I_{\{T>\tau\}}=E\left(H(1-\widetilde Z_T)\ \big|\ {\cal F}_{T-}\right)\left(1-Z_{T-}\right)^{-1}I_{\{T>\tau\}}.$$
   for any ${\cal F}_T$-measurable random variable $H$ such that the above conditional expectations exist (see Proposition 5.3 of \cite{Jeu}).
   Therefore, if assertion (b) holds, then assertion (c) follows immediately from
   (\ref{equalities22}). Conversely, if assertion (c)
    holds, then  $E^{Q_2}[ \xi I_{\{Z_{T-}<1\}} I_{\{T>\tau\}} | {\cal G}_{T-}]=0$.
     Thus, a combination of this with (\ref{equalities22}) leads to  $E\left[\xi I_{\{\widetilde{Z}_{T}<1\}} | {\cal F}_{T-}\right] (1-Z_{T-}) = 0$.
     This proves assertion (b), and the proof of the theorem is completed.
\end{proof}

 \begin{remark}\label{remarkofexplicitDefthin}
Theorem \ref{cruciallemma2singlejump} can be viewed as continuous-time version of Theorem 4.5 in \cite{choulli/deng}, and it can be generalized easily to the case of a finite number of ordered $\mathbb F$-predictable stopping times on the one hand. On the other hand, when extending this theorem to the case of general thin semimartingales, the main difficulty lies in the fact of finding a positive $\mathbb F$-local martingale, $L$ such that the density of $Q_T^{\mathbb F}$ defined in (\ref{probabilityQ'}) coincides with $L_T$ for any $\mathbb F$-predictable stopping time $T$. This difficulty remains an open problem and we are unable to see how to approach it. In contrast to $Q_T^{\mathbb F}$, the probability $Q^{\mathbb G}_T$ --given also in (\ref{probabilityQ'})-- satisfies $dQ^{\mathbb G}_T/dP={\widetilde L}^{(a)}_T/{\widetilde L}^{(a)}_{T-}$, where ${\widetilde L}^{(a)}$ is a positive $\mathbb G$-local martingale that will be described below.  To this end we need to introduce some notations and recall some results from \cite{aksamit/choulli/deng/jeanblanc3}.
\end{remark}
Throughout the rest of this subsection, we consider the following notations for any $M\in{\cal M}_{loc}(\mathbb F)$
\begin{eqnarray}
&&{\widehat M}^{(a)}:=M-M^{\tau}+(1-Z_{-})^{-1}I_{\Lbrack\tau,+\infty\Rbrack}\cdot \langle m\rangle^{\mathbb F}\in{\cal M}_{loc}(\mathbb G),\label{equa224}\\
&&W^{\mathbb G}:=\sum \ ^{p,\mathbb F}\left(I_{\{\widetilde Z=1\}}\right)I_{\Lbrack\tau,+\infty\Rbrack},\label{equa223}\\
&&K^{(a)}:={{(1-Z_{-})^2(1-\widetilde Z)^{-1}}\over{(1-Z_{-})^2+\Delta\langle m\rangle^{\mathbb F}}}I_{\Lbrack\tau,+\infty\Rbrack}\label{equa224}.
\end{eqnarray}

\noindent In the following, we recall a useful result from \cite{aksamit/choulli/deng/jeanblanc3}.

\begin{proposition}\label{propositionofACDJ2014b} The following assertions hold.\\
(a) The positive process $(1-\Delta W^{\mathbb G})^{-1}$ is $\mathbb G$-locally bounded.\\
(b) The $\mathbb G$-optional process, $K^{(a)}$, is $\widehat m^{(a)}$-integrable (with respect to Definition \ref{OIDefinition}). The resulting integral
\begin{equation}
{\widetilde L}^{(a)}:={\cal E}\left(K^{(a)}(1-\Delta W^{\mathbb G})^{-1}\odot {\widehat m}^{(a)}\right),\label{equa222}
\end{equation}
is a positive $\mathbb bG$-local martingale satisfying $[{\widetilde L}^{(a)}, \widehat M^{(a)}]\in {\cal A}_{loc}(\mathbb G)$.
\end{proposition}

\noindent In order to extend Theorem \ref{cruciallemma2singlejump} to the case of general thin semimartingales, we start by connecting the probability $Q^{\mathbb G}_T$ and ${\widetilde L}^{(a)}$ as follows.

\begin{remark}\label{remarkConnections}
Put $L^{\mathbb G}:=K^{(a)}\odot \widehat m^{(a)}$. Then, we derive
 \begin{eqnarray*}
 D^{\mathbb G}(T):&=&{{1-Z_{T-}}\over{1-\widetilde Z_{T}}}{{I_{\{T>\tau\}}}\over{P(\widetilde Z_T<1|{\cal F}_{T-})}}+I_{\{T\leq \tau\}}=
 \left(1+{{\Delta m_T}\over{1-\widetilde Z_T}}\right)I_{\{ T>\tau\}}+I_{\{ T\leq\tau\}}\\
 &=&{{1+\Delta L^{\mathbb G}-\Delta V^{\mathbb G}}\over{1-\Delta V^{\mathbb G}}}=1+\Delta {\widetilde L}^{(a)}={{{\widetilde L^{(a)}}_T}\over{{\widetilde L^{(a)}}_{T-}}}.\end{eqnarray*}
 As a result, assertions (a) and (b) of Theorem \ref{cruciallemma2singlejump} are equivalent to
 \begin{equation}\label{conclusion2}
 {\widetilde L}^{(a)} (M-M^{\tau})\ \mbox{is a}\ \mathbb G\mbox{martingale},\end{equation}
 for any single jump $\mathbb F$-martingale, $M$, with predictable jump time.
\end{remark}

\noindent Now, we are at the stage of extending Theorem \ref{cruciallemma2singlejump} to the general case of thin processes.

\begin{theorem}\label{explicitdeflatorThin4AfterTau}
Let $M$ be a thin $\mathbb F$-local martingale such that
\begin{equation}\label{Condition4AfterTau}
^{p,\mathbb F}\left(\Delta M I_{\{\widetilde Z=1>Z_{-}\}}\right)\equiv 0.\end{equation}
Then, ${\widetilde L}^{(a)}\left(M-M^{\tau}\right)$ is a $\mathbb G$-local martingale.
\end{theorem}

\begin{proof}
Thanks to It\^o's formula, it is immediate that ${\widetilde L}^{(a)}\left(M-M^{\tau}\right)$ is a $\mathbb G$-local martingale if and only if \begin{equation}\label{deflatorthin1}
X^{\mathbb G}:=M-M^{\tau}+[{\widetilde L}^{(a)},M-M^{\tau}]\end{equation}
 is a $\mathbb G$-local martingale. Since $X^{\mathbb G}$ is a $\mathbb G$-special semimartingale, hence it is enough to prove that $X^{\mathbb G}$ is a $\sigma$-martingale under $\mathbb G$. To prove this latter fact, thanks to Proposition 3.3 and Corollary 3.5 of \cite{anselstricker1994}, it is enough to prove that $\Phi\cdot X^{\mathbb G}$ is $\mathbb G$-local martingale for some $\mathbb G$-predictable process $\Phi$ such that $0<\Phi\leq 1$. Since $M$ is a thin process with predictable jump times only that we denote by $(T_n)_{n\geq 1}$, we get
 $$
 X^{\mathbb G}=\sum {\widetilde L}^{(a)}\Delta M I_{\Lbrack\tau,+\infty\Rbrack},$$
 and jumps on the sequence of stopping times $(T_n)_{n\geq 1}$ only on the one hand. On the other hand, due to Proposition \ref{propositionofACDJ2014b} (assertion (b)), we have $^{p,\mathbb G}({\widetilde L}^{(a)}\vert \Delta M\vert)I_{\Lbrack\tau,+\infty\Rbrack}<+\infty$, and hence the $\mathbb G$-predictable process
  $$
  \Phi:=\left[\sum I_{\Rbrack T_n\Lbrack} 2^{-n} +I_{\Omega\setminus(\cup_n\Rbrack T_n\Lbrack)}\right]\left(1+\ ^{p,\mathbb G}({\widetilde L}^{(a)}\vert \Delta M\vert)I_{\Lbrack\tau,+\infty\Rbrack}\right)^{-1},$$ satisfies $0<\Phi\leq 1$, $\Phi\cdot X^{\mathbb G}\in {\cal A}(\mathbb G)$, and its $\mathbb G$-compensator is given by
  $$
  (X^{\mathbb G})^{p,\mathbb G}=\sum_n \Phi\ ^{p,\mathbb G}({\widetilde L}^{(a)}\Delta M^{(n)})I_{\Lbrack\tau,+\infty\Rbrack}=0.$$
  Here $M^{(n)}:=\Delta M_{T_n}I_{\Rbrack T_n,+\infty\Rbrack}$, while the last equality follows from (\ref{conclusion2}) of Remark \ref{remarkConnections}. This proves that $\Phi\cdot X^{\mathbb G}$ is a $\mathbb G$-local martingale, and the proof of the theorem is completed.
   \end{proof}

\begin{corollary} a) If $M$ be a thin $\mathbb F$-local martingale such that $\{\Delta M\not=0\}\cap\{\widetilde Z=1>Z_{-}\}=\emptyset$, then ${\widetilde L}^{(a)}(M-M^{\tau})$ is a $\mathbb G$-local martingale.\\
b) Suppose that $S$ is thin, $\{\Delta S\not=0\}\cap\{\widetilde Z=1>Z_{-}\}=\emptyset$, and $S$ satisfies the NUPBR$(\mathbb F)$. Then $S-S^{\tau}$ satisfies the NUPBR$(\mathbb G)$.
\end{corollary}
\begin{proof} Since $S$ satisfies the NUPBR$(\mathbb F)$, then there exist an $\mathbb F$-predictable process $\phi$, a sequence of $\mathbb F$-stopping times $(T_n)_{n\geq 1}$ that increases to infinity, and a probability measure $Q_n\sim P$ on $(\Omega, {\cal F}_{T_n})$ such that
 $$
 0<\phi\leq 1,\ \ \ \ \phi\is S^{T_n}\in {\cal M}_{0,loc}(Q_n,\mathbb F).$$
Recall that for any $Q\sim P$, $\{\widetilde Z=1\}=\{\widetilde Z^Q=1\}$ where $\widetilde Z^Q_t:=Q(\tau\geq t|{\cal F}_t)$. Thus, a combination of this fact with $\{\Delta S\not=0\}\cap\{\widetilde Z=1>Z_{-}\}=\emptyset$ leads to
$$
\{\Delta (\phi\is S^{T_n})\not=0\}\cap\{\widetilde Z^{Q_n}=1>Z^{Q_n}_{-}\}=\emptyset.$$
Therefore, by applying directly Theorem \ref{explicitdeflatorThin4AfterTau} to $\phi\is S^{T_n}$ under $Q_n$, we conclude that $\phi\is S^{T_n}-(\phi\is S^{T_n})^{\tau}$ (or equivalently $S^{T_n}-S^{T_n\wedge\tau}$) satisfies the  NUPBR$(\mathbb G, Q_n)$. Hence, the corollary follows immediately from Proposition \ref{NUPBRLocalization}. This ends the proof of the corollary.
\end{proof}
%%%%%%%%%%%%%%%%%%%%%%%%%%%%%%%%%%%%%%%%%%%%%%%%%%%%%%%%%%%%%%%%%%%%%%%%%%%%%%%%%%%%%%%%%%%%%%%%%%%%%%%%%%%%%%%%%%%%%%%%%%%%%%%%%%%%%%%%%%%%%%%%%
%%%%%%%%%%%%%%%%%%%%%%%%%%%%%%%%%%%%%%%%%%%%%%%%%%%%%%%%%%%%%%%%%%%%%%%%%%%%%%%%%%%%%%%%%%%%%%%%%%%%%%%%%%%%%%%%%%%%%%%%%%%%%%%%%%%%%%%%%%%%%%%%%
%%%%%%%%%%%%%%%%%%%%%%%%%%%%%%%%%%%%%%%%%%%%%%%%%%%%%%%%%%%%%%%%%%%%%%%%%%%
%%%%%%%%%%%%%%%%%%%%%%%%%%%%%%%%%%%%%%%%%%%%%%%%%%%%%%%%%%%%%%%%%%%%%%%%%%%%
%%%%%%%%%%%%%%%%%%%%%%%%%%%%%%%%%%%%%%%%%%%%%%%%%%%%%%%%%%%%%%%%%%%%%%%%%%%%%%%%%%%%%%%%%%%%%%%%%%%%%%
\section{Proofs of Theorems \ref{main3} and \ref{cruciallemma2}}\label{proofs}

In this section, we prove Theorems \ref{main3} and \ref{cruciallemma2}. These  proofs are not technical, but are long instead.
%%%%%%%%%%%%%%%%%%%%%%%%%%%%%%%%%%%%%%%%%%%%%%%%%%%%%%%%%%%%%%%%%%%%%%%%%%%%%%%%%%%%%%%%%%%%%%%%%%%%%
%%%%%%%%%%%%%%%%%%%%%%%%%%%%%%%%%%%%%%%%%%%%%%%%%%%%%%%%%%%%%%%%%%%%%%%%%%%%%%%%%%%%%%%%%%%%%%%%%%%%%

\subsection{Proof of Theorem \ref{main3}} \label{proofmain3}
The proof is achieved in four steps, where we prove ${\rm{(c)}}\Longleftrightarrow$(d), (d)$ \Longleftrightarrow {\rm{(b)}}$, (a)$\Longrightarrow {\rm{(c)}}$, and ${\rm{(b)}}\Longrightarrow {\rm{(a)}}$ respectively.\\
%%%%%%%%%%%%%%%%%%%%%%%%%%%%%%%%%%%%%%%%%%%%%%%%%%%%%%%%%%%%%%%%%%%%%%%%%%%%
\noindent{\bf Step 1:} In this step, we  prove
${\rm{(c)}}\Longleftrightarrow {\rm{(d)}}$. Since $S$ is a single
jump process with predictable jump time $T$, then it is easy to
see that $S$ satisfies the NUPBR$(R)$, for some probability $R$, is
equivalent to the fact that $I_A S$ and $I_{A^c} S$ satisfies
NUPBR$(R)$ for any ${\cal F}_{T-}$-measurable event $A$. Hence, it
is enough to prove the equivalence between assertions \rm{(d)}
and \rm{(c)} separately on the events $\{Z_{T-} = 0 \}$ and
$\{Z_{T-} > 0 \}$. Since $\{Z_{T-}=0\}\subset\{\widetilde Z_T=0\}$
and $E(\widetilde Z_T|{\cal F}_{T-})=Z_{T-}$ on $\{T<+\infty\}$,
by putting $\Gamma_0:=\Bigl\{ P(\widetilde{Z}_T > 0 \big| {\cal
F}_{T-}) =0\ \&\ T<+\infty\Bigr\}$, we derive
$$E\left(Z_{T-}I_{\Gamma_0\cap\{T<+\infty\}}\right)=E\left({\widetilde
Z}_{T}I_{\Gamma_0\cap\{T<+\infty\}}\right)=0,$$ and
  $$\begin{array}{lll}
  0=P\left(\{Z_{T-}=0\}\cap\{\widetilde Z_{T}>0\}\cap\{T<+\infty\}\right)\\
  \\
  \hskip 0.35cm=E\left(I_{\{ Z_{T-}=0\}\cap\{T<+\infty\}}P\left(\widetilde Z_T>0|{\cal F}_{T-}\right)\right).\end{array}
  $$
These equalities imply that on $\{T<+\infty\}$, $P-a.s.$, we
have
  \begin{equation}\label{equalitybetweensets}
  \{Z_{T-} = 0 \}=\Gamma_0   \subset \{\widetilde{Z}_T =0\}.
  \end{equation}
  Thus, on the set $\{T<+\infty\}\cap\Gamma_0$, the three probabilities $P$, $Q_T$ and $\widetilde{Q}_T$ coincide, and the equivalence between assertions \rm{(c)}
   and \rm{(d)} is obvious. On the set $\{T<+\infty\ \&\  P[\widetilde{Z}_T >0 | {\cal F}_{T-}] >0\}$, one has $\widetilde{Q}_T \sim Q_T$, and the equivalence between \rm{(c)} and \rm{(d)} is also obvious. This achieves this first step.\\
   \noindent{\bf Step 2:} This step proves
\rm{(d)}$\Longleftrightarrow$ \rm{(b)}. Thanks to $\{Z_{T-}=0\}\subset\{\widetilde Z_T=0\}$, we deduce that on $\{ Z_{T-}=0\}$,
$\widetilde S\equiv S\equiv 0$ and $Q_T$ coincides with $P$ as
well. Hence, the equivalence between assertions \rm{(d)} and
\rm{(b)} is obvious for this case. Thus, it is enough to prove the
equivalence between these assertions on $\{T<+\infty\ \&\
P(\widetilde{Z}_T >0 | {\cal F}_{T-}) >0\}$.\\ Assume that
\rm{(d)} holds. Then, there exists   an ${\cal F}_T$-measurable
random variable, $Y$, such that $Y>0$ $Q_T-a.s.$  and on
$\{T<+\infty\}$, we have
 $$E^{Q_T}(Y| {\cal F}_{T-}) =1,\ \ E^{Q_T}(Y|\xi| | {\cal F}_{T-}) <+\infty,\ \ \&\  \ \ E^{Q_T}(Y\xi I_{\{\widetilde{Z}_T >0\}} | {\cal F}_{T-}) =0 .$$
 Since $Y>0$ on $\{\widetilde Z_T>0\}$, by putting
  $$
  Y_1 := YI_{\{\widetilde{Z}_T >0\}} + I_{\{\widetilde{Z}_T =0\}} \ \ \mbox{and}\ \ \widetilde{Y}_1 := \frac{Y_1}{E[Y_1 | {\cal F}_{T-}]} ,
  $$
  it is easy to check that $Y_1>0$, $\widetilde Y_1>0$,
  \begin{eqnarray*}
    E\left[\widetilde{Y}_1 | {\cal F}_{T-}\right] = 1 \mbox{ and } E\left[\widetilde{Y}_1 \xi I_{\{\widetilde{Z}_T >0\}} | {\cal F}_{T-}\right] = \frac{E\left[Y \xi I_{\{\widetilde{Z}_T >0\}} | {\cal F}_{T-}\right]}{E[Y_1 | {\cal F}_{T-}]} = 0.
  \end{eqnarray*}
   Therefore, $\widetilde{S}$ is  a martingale under $R:=\widetilde Y_1\cdot P\sim P$, and hence $\widetilde S$ satisfies NUPBR$(\mathbb F)$.
   This ends the proof of (a)$\Rightarrow$(b). To prove the reverse sense, we suppose  that assertion
   \rm{(b)} holds. Then, there exists $0<Y\in L^0({\cal F}_T)$,
   such that $E[Y |\xi| I_{\{\widetilde{Z}_T >0\}} | {\cal F}_{T-}] <+\infty$,
    $E[Y| {\cal F}_{T-}] =1$ and
    $ E[Y\xi I_{\{\widetilde{Z}_T >0\}} | {\cal F}_{T-}] = 0$ on $\{Z_{T-}>0\}$.
    Then, consider
    $$ Y_2 :=
\frac{Y I_{\{\widetilde{Z}_T >0\}}
    P(\widetilde{Z}_T >0 | {\cal F}_{T-})}{E[Y I_{\{\widetilde{Z}_T >0\}} | {\cal F}_{T-}]}+
    I_{\{\widetilde{Z}_T=0\}}
  $$
  Then it is easy to verify that $Y_2>0\ \ Q_T-a.s.$,
  $$
    E^{Q_T}\left(Y_2 | {\cal F}_{T-}\right) = 1, \ \ \ \ \ \mbox{and}\ \ \ \ \
    E^{Q_T}\left(Y_2 \xi I_{\{Z_{T-}>0\}}| {\cal F}_{T-}\right) = \frac{E\left[Y \xi I_{\{\widetilde{Z}_T >0\}} | {\cal F}_{T-}\right]}{E[Y I_{\{\widetilde{Z}_T >0\}} | {\cal F}_{T-}]} =0.
  $$
  This proves assertion \rm{(d)}, and the proof of \rm{(d)}$\Longleftrightarrow$\rm{(b)} is achieved.\\
\noindent{\bf Step 3:} Herein, we prove \rm{(a)} $\Rightarrow$
\rm{(c)}. Suppose that $S^{\tau}$ satisfies NUPBR$(\mathbb G)$.
Then there exists a positive ${\cal G}_T$-measurable random
variable $Y^{\mathbb G}$ such that $E[\xi Y^{\mathbb G} I_{\{T\leq
\tau\}} | {\cal G}_{T-}] = 0$ on $\{T<+\infty\}$. Due to Lemma
\ref{lemma:predsetFG}--\rm{(a)}, we deduce the existence of two positive ${\cal F}_T$-measurable variables $Y^{\mathbb F}_1 $  and $Y^{\mathbb F}_2 $ such
that $Y^{\mathbb G} I_{\{T\leq \tau\}} = Y^{\mathbb F}_1 I_{\{T<
\tau\}}+Y^{\mathbb F}_2 I_{\{T=\tau\}} $. Then, on $\{T<+\infty\}$, we obtain
  \begin{eqnarray*}
    0&=& E[\xi Y^{\mathbb G} I_{\{T \leq \tau\}} | {\cal G}_{T-}] = E[\xi (Y^{\mathbb F}_1 Z_T+(Z_T-\widetilde{Z}_T)Y^{\mathbb F}_2 | {\cal F}_{T-}]  \frac{I_{\{T\leq \tau\}}}{Z_{T-}}.
  \end{eqnarray*}
  Therefore, by taking conditional expectation in the above equality and putting $${\widetilde Y}:= Y^{\mathbb F}_1 {{Z_T}\over{\widetilde{Z}_T}}I_{\{\widetilde Z_T>0\}}+({{Z_T}\over{\widetilde{Z}_T}}-1)I_{\{\widetilde Z_T>0\}}Y^{\mathbb F}_2+I_{\{\widetilde Z_T=0\}}>0,$$ we get $$0=E[\xi \widetilde{Y} \frac{\widetilde{Z}_T}{ Z_{T-}} I_{\{Z_{T-}>0\}}| {\cal F}_{T-}] = E^{\widetilde{Q}_T}[\xi \widetilde{Y}| {\cal F}_{T-}]I_{\{Z_{T-}>0\}}=E^{\widetilde{Q}_T}[S_T  \widetilde{Y}| {\cal F}_{T-}].$$ This proves that assertion (d) holds and the proof of \rm{(a)}$\Rightarrow $\rm{(d)} is achieved.\\
\noindent{\bf Step 4:} This last step proves \rm{(b)}$\Rightarrow
$\rm{(a)}. Suppose that $\widetilde{S}$ satisfies   NUPBR$(\mathbb
F)$. Then, there exists $Y\in L^1({\cal F}_T)$ such that on
$\{T<+\infty\}$ we have
  \begin{eqnarray*}
    E[Y|{\cal F}_{T-}] = 1, \ \ Y>0,\ \ \  \ E[Y|\xi|I_{\{\widetilde{Z}_T>0\}} | {\cal F}_{T-}] <+\infty,\ \ P-a.s.
  \end{eqnarray*}
  and
  \begin{eqnarray*}
    E[Y \xi I_{\{\widetilde{Z}_T>0\}} | {\cal F}_{T-}] = 0.
  \end{eqnarray*}
  Then by putting $R:=Y\cdot P\sim P$, we deduce that $\widetilde S$ is an $(\mathbb F, R)$-martingale and $\Delta S I_{\{\widetilde Z=0\}}\equiv 0$. As a result, assertions (a) follows from direct application of Proposition \ref{cruciallemma1} to $M:=\widetilde S$ under $R\sim P$ (it is easy to see that (\ref{zeroequationbeforetau}) holds for $(\widetilde S, R)$, i.e.  $E^R({\widetilde S}_T I_{\{ \widetilde Z_T=0\}}|{\cal F}_{T-})=0$). This ends the fourth step and the proof of the theorem is completed.\qed

%%%%%%%%%%%%%%%%%%%%%%%%%%%%%%%%%%%%%%%%%%%%%%%%%%%%%%%%%%%%%%%%%%%%%%%%%%%%%%%%%%%%%%%%%%%%%%%%%%%%%%%%%%%%%%%%
%%%%%%%%%%%%%%%%%%%%%%%%%%%%%%%%%%%%%%%%%%%%%%%%%%%%%%%%%%%%%%%%%%%%%%%%%%%%%%%%%%%%%%%%%%%%%%%%%%%%%%%%%%%%%

\subsection{Proof of Theorem \ref{cruciallemma2}}\label{proofofcruciallemma2} Due to $\{Z_{T-} =1\} = \{P(\widetilde{Z}_T<1 | {\cal F}_{T-}) = 0\} \subset \{\widetilde{Z}_T = 1\}$, it is obvious that ${\widetilde Q}'_T \sim Q'_T \ll P.$ Thus, (b)$\Longleftrightarrow $(c) follows immediately. Thus, the remaining part of the proof consists of three steps, where (c)$\Longrightarrow $(d), (d)$\Longrightarrow $(a) and (a)$\Longrightarrow $(b) are proven respectively.\\
\noindent{\bf  Step 1:}(c)$\Rightarrow $(d). Suppose (c) holds. Then, there exists an ${\cal F}_T$-measurable random variable $Y_T>0,\ Q'_T$-a.s. such that $E^{Q'_T}[S_TY_T | {\cal F}_{T-}] = 0$, or equivalently $$E[\xi Y_TI_{\{ \widetilde Z_T<1\}} | {\cal F}_{T-}]I_{\{ Z_{T-}<1\}} = 0\ \ \mbox{and}\ \ E[\xi Y_{T} | {\cal F}_{T-}]I_{\{ Z_{T-}=1\}} =0.$$ Since, on the set $\{ Z_{T-}=1\}$, $\widetilde{S} \equiv 0$, it is enough to focus on the part corresponding to $\{ Z_{T-}<1\}$.  Put
  $$\widetilde Y_T:= Y_TI_{\{ \widetilde Z_T<1\}} + I_{\{ \widetilde Z_T=1\}}\ \ \ \mbox{ and}\ \  Q_1:= {\widetilde Y}_T/E({\widetilde Y}_T|{\cal F}_{T-})\cdot P \sim P.$$
   Then, we derive that $E^{Q_1}[\xi  I_{\{ \widetilde Z_T<1\}} | {\cal F}_{T-}] = 0$. Therefore, we conclude that $\widetilde{S}$ is a ($Q_1,\mathbb F)$-martingale, and hence assertion (d) follows.\\
\noindent{\bf  Step 2:} (d)$\Rightarrow$ (a). Since $\widetilde{S}$ satisfies NUPBR($\mathbb F$), then there exists an ${\cal F}_T$-measurable $Y_3>0$ such that $E[Y_3 \xi  I_{\{ \widetilde Z_T<1\}}|{\cal F}_{T-}] = 0$. Put $Q_3:= Y_3/E(Y_3|{\cal F}_{T-})\cdot P\sim P$ and remark that $\{ \widetilde Z_T<1\} = \{ \widetilde Z^{Q_3}_T<1\}$, where $\widetilde Z^{Q_3}_t:= Q_3(\tau \geq t |{\cal F}_t)$. Therefore, a direct application of Theorem \ref{cruciallemma2singlejump} under $Q_3$, we conclude that $S-S^\tau = \widetilde{S} - \widetilde{S}^\tau$ satisfies NUPBR$(\mathbb G)$.\\
\noindent{\bf  Step 3:} (a)$\Rightarrow$ (b). Suppose $S-S^{\tau}$ satisfies NUPBR$(\mathbb G)$. There exists a ${\cal G}_T$-measurable $Y^{\mathbb G} >0$ such that $E[XY^{\mathbb G} I_{\{T>\tau\}} | {\cal G}_{T-}] = 0.$ Then, thanks to Proposition ??, we deduce the existence of a positive ${\cal F}_T$-measurable $\overline{Y}^{\mathbb F} $ such that $Y^{\mathbb G} I_{\{T>\tau\}} =\overline{Y}^{\mathbb F} I_{\{T>\tau\}}$. Then, we calculate
  \begin{eqnarray*}
    0&=& E[\xi  Y^{\mathbb G} I_{\{T>\tau\}} | {\cal G}_{T-}] = E[\xi  Y^{\mathbb F} (1-\widetilde{Z}_T) | {\cal F}_{T-}]  \frac{I_{\{T>\tau\}}}{1 - Z_{T-}}\nonumber\\
    &=&E^{\widetilde Q'(T)}\left(XY^{\mathbb F}\big|\ {\cal F}_{T-}\right)I_{\{T>\tau\}}.
  \end{eqnarray*}
  Therefore, by taking conditional expectation and using the fact that the support of $\widetilde Q'(T)$ is included in $\{ Z_{T-}<1\}$, we obtain  $$(1-Z_{T-})E^{\widetilde Q'(T)}[\xi  Y^{\mathbb F} \big| {\cal F}_{T-}]=0,\ \mbox{or equivalently}\ \ E^{\widetilde Q'(T)}[S_{T}Y^{\mathbb F} \big| {\cal F}_{T-}]=0 \ \ P-a.s.$$ This proves assertion (b), and the proof of the theorem is achieved.\qed

%%%%%%%%%%%%%%%%%%%%%%%%%%%%%%%%%%%%%%%%%%%%%%%%%%%%%%%%%%%%%%%%%%%%%%%%%%%%%%%%%%%%%%%%%%%%%%%%%%%%%%%%%%%%%
%%%%%%%%%%%%%%%%%%%%%%%%%%%%%%%%%%%%%%%%%%%%%%%%%%%%%%%%%%%%%%%%%%%%%%%%%%%%%%%%%%%%%%%%%%%%%%%%%%%%%%%%%%%%%
%%%%%%%%%%%%%%%%%%%%%%%%%%%%%%%%%%%%%%%%%%%%%%%%%%%%%%%%%%%%%%%%%%%%%%%%%%%%%%%%%%%%%%%%%%%%
%%%%%%%%%%%%%%%%%%%%%%%%%%%%%%%%%%%%%%%%%%%%%%%%%%%%%%%%%%%%%%%%%%%%%%%%%%%%%%%%%%%%%%%%%%%%%%%%
\section{Proof of Theorems \ref{main4} and \ref{MainTheoremMultiJumps}}\label{Proof4theoremsNotTech}
This section is devoted to the proofs of Theorems \ref{main3} and \ref{MainTheoremMultiJumps}. These proofs are technical and require some notations on random measures and semimartingale characteristics. For any filtration
$\mathbb H$,   we denote
$$\label{sigmaFields}
\widetilde {\cal O}(\mathbb H):={\cal O}(\mathbb H)\otimes {\cal
B}({\mathbb R}^d),\ \ \ \ \ \widetilde{\cal P}(\mathbb H):= {\cal
P}(\mathbb H)\otimes {\cal B}({\mathbb R}^d),
$$
where ${\cal B}({\mathbb R}^d)$ is the Borel $\sigma$-field on
${\mathbb R}^d$.  To a c\`adl\`ag $\mathbb H$-adapted process $X$, we associate the following optional random measure $\mu_X$ defined by

\begin{equation}\label{mesuresauts}
\mu_X(dt,dx):=\sum_{u>0} I_{\{\Delta X_u \neq 0\}}\delta_{(u,\Delta
X_u)}(dt,dx)\,.\end{equation}

\noindent For a product-measurable functional
$W\geq 0$ on $\Omega\times[0,+\infty[\times{\mathbb R}^d$, we
denote $W\star\mu_X$ (or sometimes, with abuse of notation
$W(x)\star\mu_X$) the process
\begin{equation}\label{Wstarmu}
(W\star\mu_X)_t:=\int_0^t \int_{{\mathbb R}^d-\{0\}}
W(u,x)\mu_X(du,dx)=\sum_{0<u\leq t} W(u,\Delta X_u) I_{\{ \Delta
X_u\not=0\}}.\end{equation}

\begin{definition}\label{g1} Consider a c\`adl\`ag $\mathbb H$-adapted process $X$, and its optional random measure $\mu_X$. \\
(a) We denote by ${\cal
G}^1_{loc}(\mu_X, \mathbb H)$, the set of all $\widetilde {\cal
P}(\mathbb H)$-measurable functions, $W$, such that
$$\left[\sum_{t\leq\cdot} \left(W(t,\Delta S_t)I_{\{\Delta S_t\not=0\}}-\int W_t(x)\nu_X(\{t\},dx)\right)^2\right]^{1/2} \in {\cal A}_{loc}^+(\mathbb H).$$
(b) The set ${\cal
H}^1_{loc}(\mu_X,\mathbb H)$) is the set of all $\widetilde
{\cal O}(\mathbb H)$-measurable functions,  $W$, such that $(W^2\star\mu_X)^{1/2}\in {\cal A}^+_{loc}(\mathbb H).$
\end{definition}
\noindent Also on
$\Omega\times[0,+\infty[\times{\mathbb R}^d$, we define the
measure $M^P_{\mu_X}:=P\otimes\mu_X$ by $$\int W
dM^P_{\mu_X}:=E\left[(W\star\mu_X)_\infty\right],$$ (when the expectation is well defined). The conditional ``expectation" given $
\widetilde{\cal P}(\mathbb H)$ of a product-measurable functional
$W$, denoted by $M^P_{\mu_X}(W|\widetilde{\cal P}(\mathbb H))$, is the unique $ \widetilde{\cal P}(\mathbb H)$-measurable
functional $\widetilde W$ satisfying
$$
E\left[(W I_{\Sigma}\star\mu_X)_\infty \right]=E\left[({\widetilde W}
I_{\Sigma}\star\mu_X)_\infty \right],\ \ \ \mbox{for all}\
\Sigma\in\widetilde{\cal P}(\mathbb H).$$

\noindent When $X=S$, for the sake of simplicity, we denote $\mu:=\mu_S$. Then, the $\mathbb F$-canonical decomposition of $S$ is
\begin{equation}\label{representationS}
S=S_0+h\star(\mu-\nu)+b\cdot A+(x-h)\star\mu,\end{equation}
where $h$, defined as $h(x):=xI_{\{ \vert x\vert\leq 1\}}$, is the
truncation function. We associate to $\mu$ defined in (\ref{Wstarmu}) when $X=S$,  its predictable
compensator random measure $\nu$. A direct application of  Theorem A.1 in \cite{aksamit/choulli/deng/jeanblanc2} (see also Theorem 3.75 in \cite{Jacod} (page 103), or Lemma 4.24 in \cite{JacodShirayev} (Chap III)), to the martingale $m$ defined in (\ref{processmZ}), leads
to the existence of a local martingale $m^\bot$ as well as $f_m\in{\cal
G}^1_{loc}(\mu,\mathbb F )$, $g_m\in{\cal H}^1_{loc}(\mu,\mathbb F
)$ and $\beta_m\in L( {S^c})$ such that
\begin{eqnarray}\label{decompositionofm}
 m=\beta_m\is S^c+f_m\star(\mu -\nu)+ g_m\star \mu +  m^ \bot.
\end{eqnarray}

\noindent The corresponding canonical decomposition of $S^{\tau}$
under $\mathbb G$ is given by
\begin{equation}\label{StauDecomposition}
S^{\tau}=S_0+h\star(\mu^{\mathbb G}_b-\nu^{\mathbb G}_b)+h{{f_m}\over{Z_{-}}}I_{\Lbrack 0,\tau\Lbrack}\star\nu+b\is A^{\tau}+(x-h)\star\mu^{\mathbb G}_b\end{equation}
where $(\beta_m, f_m)$  is given by (\ref{decompositionofm}) and $\mu^{\mathbb G}_b$ and $\nu^{\mathbb G}_b$ are given by
 \begin{equation}\label{muGnuGbeforeTau}
 \mu^{\mathbb G}_b(dt,dx):=I_{\Rbrack 0,\tau\Lbrack}(t)\mu(dt,dx),\ \ \nu^{\mathbb G}_b(dt,dx):=(1+Z_{-}^{-1}f_m)I_{\Rbrack 0,\tau\Lbrack}(t)\nu(dt,dx).\end{equation}

\subsection{Proof of Theorem \ref{main4}} \label{proofmain4} This proof consists of four steps, where we prove (b)$\Longleftrightarrow$(c), (b)$\Longrightarrow $(a), and (a)$\Longrightarrow $(b) respectively. Only the last step is technically involved.\\
\noindent {\bf Step 1:} Here, we prove (b)$\Longleftrightarrow$(c). Remark that (c)$\Longrightarrow$(b) follows immediately from Lemma \ref{LY}. Suppose that assertion (b) holds, and consider the following $\mathbb F$-predictable process
$$\varphi:=\left[1+\ ^{p,\mathbb F}\left(Y\vert \Delta S\vert I_{\{\widetilde Z>0\}}\right)\right]^{-1}\left[I_{\Omega\setminus(\cup_n \Rbrack T_n\Lbrack)}+\sum 2^{-n}I_{\Rbrack T_n\Lbrack}\right],$$
where $(T_n)_n$ a sequence of $\mathbb F$-predictable stopping times such that $\{\Delta S\not=0\}\subset\displaystyle\bigcup_{n=1}^{+\infty}\Rbrack T_n\Lbrack$. Then, it is easy to see that the process $$X:=Y_{-}\varphi\cdot S^{(0)}+[\varphi\cdot S^{(0)}, Y]=\sum Y\varphi\Delta S I_{\{\widetilde Z<1\ \&\ Z_{-}\geq\delta\}}$$has an integrable variation and its $\mathbb F$-compensator is given by (due to the fact it is a pure jump process with finite variation and it jumps on predictable stopping times only)
$$
X^{p,\mathbb F}=\sum\ ^{p,\mathbb F}\left(Y\varphi\Delta S I_{\{\widetilde Z>0\}}\right)I_{\{Z_{-}\geq\delta\}}\equiv 0.$$
Thus, $Y(\varphi\cdot S^{(0)})$ is an $\mathbb F$-local martingale, and $S^{(0)}$ satisfies the NUPBR$(\mathbb F)$. This ends the proof of (b)$\Longleftrightarrow$(c).\\
\noindent{\bf Step 2:}  Here, we prove (b)$\Rightarrow$ (a).
Suppose that assertion (b) holds, and consider a sequence of
$\mathbb F$-stopping times $(\tau_n)_n$ that increases to infinity
such that $Y^{\tau_n}$ is an $\mathbb F$-martingale, and
$Q_n:=Y_{\tau_n}/Y_0\cdot P\sim P$. Then, (\ref{mainassumptionbeforetau}) implies that $(S^{(0)})^{\sigma_n}$ is a $Q_n$-local martingale and satisfies (\ref{Condition1}) under $Q_n$ due to
\begin{equation}\label{Zzerosetbeforetau}
\{\widetilde{Z}^Q_T=0\} = \{\widetilde{Z}_T = 0\},\ \mbox{for any}\  Q\sim P\ \mbox{and any ${\mathbb F}$-stopping time}\ T,\end{equation}
 where $\widetilde{Z}^Q_t := Q[\tau \geq t| {\cal F}_t]$. This follows from \begin{eqnarray*}
    E\left[\widetilde{Z}_T I_{\{\widetilde{Z}^Q_T=0\}}\right] = E\left[ I_{\{\tau \geq T\}} I_{\{\widetilde{Z}^Q_T=0\}}\right] = 0,
  \end{eqnarray*}
(which implies $\{\widetilde{Z}^Q_T=0\} \subset \{\widetilde{Z}_T
= 0\}$) and the symmetric role of $Q$ and $P$.\\ Thus, a direct application of Theorem \ref{generalThinDeflators} to $\left((S^{(0)})^{\sigma_n},Q_n\right)$ leads to the NUPBR$(\mathbb G, Q_n)$ of $(S^{(0)})^{\sigma_n\wedge\tau}=\left(I_{\{Z_{-}\geq \delta\}}\cdot S\right)^{\sigma_n\wedge\tau}$. Thanks to Proposition \ref{NUPBRLocalization}, this implies the NUPBR$(\mathbb G)$ of $I_{\{Z_{-}\geq \delta\}}\cdot S$ for any $\delta>0$. Since $Z_{-}^{-1}I_{\Rbrack0,\tau\Lbrack}$ is $\mathbb G$-locally bounded, there exists a family of $\mathbb G$-stopping times $\tau_{\delta}$ that increases to infinity when $\delta$ decreases to zero, and $\Rbrack 0, \tau\wedge\tau_{\delta}\Lbrack\subset\{Z_{-}\geq \delta\}$. Therefore, we conclude that $S^{\tau\wedge\tau_{\delta}}$ satisfies the NUPBR$(\mathbb G)$. Hence, again Proposition \ref{NUPBRLocalization} implies finally that $S^{\tau}$ satisfies the NUPBR$(\mathbb G)$. This ends the second part. \\
\noindent {\bf Step 3:} In this step, we focus on proving (a)$\Rightarrow $(b).  Suppose that
$S^\tau$ satisfies NUPBR($\mathbb G$). Then, there exists a
$\sigma$-martingale density under $\mathbb G$, for $I_{\{Z_{-}\geq
\delta\}}\is S^{\tau},$ ($\delta>0$), that we denote by $
D^\mathbb G$. Then, from a direct application of Theorem \ref{theosigmadensityiff},
  we deduce the existence of a positive $\widetilde{\cal P}(\mathbb G)$-measurable functional, $f^{\mathbb G}\in{\cal G}^1_{loc}(\mu^{\mathbb G}_b,\mathbb G)$, such that $D^\mathbb G :={\cal E}(N^{\mathbb G})>0$, with
$$
  N^{\mathbb G}:=W^{\mathbb G}\star (\mu^{\mathbb G} - \nu^{\mathbb G}),  \ W^{\mathbb G}:= f^{\mathbb G}-1 + \frac{\widehat{f}^{\mathbb G} - a^{\mathbb G}}{1 - a^{\mathbb G}}I_{\{a^{\mathbb G} <1\}},
 $$
where $\nu^{\mathbb G}$ was defined in (\ref{muGnuGbeforeTau}),
and, introducing $f_m$ defined in (\ref{decompositionofm})
  \begin{eqnarray}\label{mgequationbeforetau}
    xf^{\mathbb G}I_{\{Z_{-}\geq\delta\}}\star \nu^{\mathbb G}  = xf^{\mathbb G}\left(1 + \frac{f_m}{Z_{-}}\right)I_{\Lbrack 0,\tau\Lbrack}I_{\{Z_{-}\geq\delta\}}\star  \nu \equiv 0.
  \end{eqnarray}
Thanks to Lemma \ref{lemma:predsetFG}, we conclude the existence of a positive
  $\widetilde{\cal P}(\mathbb F)$-measurable functional, $f$,
  such that $f^{\mathbb G}I_{\Lbrack 0,\tau\Lbrack} = fI_{\Lbrack 0,\tau\Lbrack}$.
  Thus, (\ref{mgequationbeforetau}) becomes
 $$ U^{(b)}:=xf\left(1 + \frac{f_m}{Z_{-}}\right)I_{\Lbrack 0,\tau\Lbrack}I_{\{ Z_{-}>0\}}\star  \nu \equiv 0.
$$
 Introduce the following notations
  \begin{equation}\label{eq:cruYzero}\left\{\begin{array}{llll}
   \mu_0 := I_{\{\widetilde{Z}>0\ \&\ Z_{-} \geq \delta\}} \cdot \mu , \ \ \nu_0 := h_0I_{\{ Z_{-} \geq \delta\}}\cdot \nu, \   h_0:= M^P_{\mu}\left(I_{\{\widetilde{Z}>0\}} | \widetilde{\cal P}\right), \nonumber \\
    g := \frac{f(1 + \frac{f_m}{Z_{-}})}{h_0} I_{\left\{h_0>0\right\}} + I_{\left\{h_0=0\right\}}, \ \ a_0(t):= \nu_0(\{t\}, {\mathbb{R}^d}),
    \end{array}\right.\end{equation}
  and assume that
  \begin{equation}\label{mainassumtpionbeforetaubis}
  \sqrt{(g-1)^2\star\mu_0}\in {\cal A}^+_{loc}(\mathbb F).\end{equation}
  Then, thanks to Lemma \ref{boundednessofuhat}, we deduce that $W:=(g-1)/(1- a^0 + \widehat{g})\in {\cal G}^1_{loc}(\mu_0,\mathbb F)$, and the local martingales
  \begin{equation}\label{Nzerobeforetau}
  N^{(0)}:= \frac{g-1}{1 - a^0 + \widehat{g}}\star(\mu_0 - \nu_0), \ \ Y^{(0)} := {\cal E}(N^{(0)}),\end{equation}
  are well defined satisfying $1 + \Delta N^{(0)} > 0$, $[N^{(0)},S]\in{\cal A}(\mathbb F)$,  and on $\{ Z_{-}>0\}$ we have
  \begin{eqnarray*}
    {{^{p,\mathbb F}\left(Y^{(0)}\Delta S I_{\{\widetilde{Z}>0\}}\right)}\over{Y^{(0)}_{-}}} &=& \ ^{p,\mathbb F}\left((1+\Delta N^{(0)})\Delta S I_{\{\widetilde{Z}>0\}}\right) =\  ^{p,\mathbb F}\left(\frac{g}{1 - a^0 + \widehat{g}}\Delta S I_{\{\widetilde{Z}>0\}}\right)\\
     &=& \Delta \frac{gxh_0}{1 - a^0 + \widehat{g}}\star \nu = \Delta\frac{ xf(1 + f_m/{Z_{-}})}{1 - a^0 + \widehat{g}}I_{\{Z_{-}>0\}}\star \nu\\
     &=&\frac{ ^{p,\mathbb F}\left(\Delta U^{(b)}\right)}{1 - a^0 + \widehat{g}} \equiv 0.
  \end{eqnarray*}
  This proves that assertion (b) holds under the assumption (\ref{mainassumtpionbeforetaubis}). The remaining part of the proof will show that this assumption holds always. To this end, we start by noticing that on the set $ \left\{h_0>0\right\}$,
  \begin{eqnarray*}
    g-1 &=& \frac{f(1 + \frac{f_m}{Z_{-}})}{h_0} - 1 =
     \frac{(f-1)(1 + \frac{f_m}{Z_{-}})}{h_0} + \frac{f_m}{Z_{-}h_0} + \frac{M^P_{\mu}\left(I_{\{\widetilde{Z}=0\}} | \widetilde{\cal P}\right)}{h_0}\nonumber \\
    &:=&  \frac{(f-1)(1 + \frac{f_m}{Z_{-}})}{h_0} + \frac{M^P_{\mu}\left(\Delta m I_{\{\widetilde{Z}>0\}} | \widetilde{\cal P}\right)}{Z_{-}h_0}=:g_1 + \frac{M^P_{\mu}\left(\Delta m I_{\{\widetilde{Z}>0\}} | \widetilde{\cal P}\right)}{Z_{-}h_0}.
  \end{eqnarray*}
  Since $\left((f-1)^2 I_{\Lbrack 0,\tau \Lbrack}\star \mu \right)^{1/2}  \in {\cal A}^+_{loc}(\mathbb{G})$,
  then due to Proposition \ref{prop:alocundergf}--(e)
  \begin{eqnarray*}
    \sqrt{(f-1)^2 I_{\{Z_{-}\geq \delta\}}\star (\widetilde{Z}\cdot \mu)}\in {\cal A}^+_{loc}(\mathbb F),\ \ \ \mbox{for any}\ \ \delta>0.\end{eqnarray*}
    Then, a direct application of Proposition \ref{prop:alocundergf}--(a), for any $\delta>0$, we have
    \begin{eqnarray*} (f-1)^2 I_{\{\vert f-1\vert\leq \alpha\ \&\ Z_{-}\geq \delta\}}\star (\widetilde{Z}\cdot \mu),\ \  |f-1|I_{\{|f-1|> \alpha\ \&\ Z_{-}\geq\delta\}}\star (\widetilde{Z}
    \cdot \mu) \in {\cal A}^+_{loc}(\mathbb F).
  \end{eqnarray*}
  By stopping, without loss of generality, we assume these two processes and $[m,m]$ belong to ${\cal A}^+(\mathbb F)$. Remark that $Z_{-} + f_m = M^P_{\mu}\left(\widetilde{Z}| \widetilde{\cal P}\right) \leq M^P_{\mu}\left(I_{\{\widetilde{Z}>0\}} | \widetilde{\cal P}\right) = h_0$ that follows from $\widetilde{Z} \leq I_{\{\widetilde{Z}>0\}}$. Therefore, we derive
  \begin{eqnarray*}
    E\left[g_1^2I_{\{|f-1|\leq \alpha\}}\star \mu_0(\infty)\right] &=& E\left[\frac{(f-1)^2(1 + \frac{f_m}{Z_{-}})^2}{h_0^2}I_{\{|f-1|\leq \alpha\}}\star \mu_0(\infty)\right] \nonumber \\
    &=& E\left[\frac{(f-1)^2(1 + \frac{f_m}{Z_{-}})^2}{h_0^2}I_{\{|f-1|\leq \alpha\}} \star \nu_0(\infty)\right] \nonumber \\
    &\leq& \delta^{-2}E\left[(f-1)^2( Z_{-} + f_m) I_{\{|f-1|\leq \alpha\ \&\ Z_{-}\geq \delta\}} \star \nu(\infty)\right] \nonumber \\
    &=& \delta^{-2}E\left[(f-1)^2 I_{\{|f-1|\leq \alpha\}} \star (\widetilde{Z}I_{\{Z_{-}\geq \delta \}}
    \cdot \mu)(\infty)\right]<+\infty, \nonumber
  \end{eqnarray*}
 and
  \begin{eqnarray*}
    E\left[g_1I_{\{|f-1|> \alpha\}}\star \mu_0(\infty)\right] &=& E\left[\frac{|f-1|(1 + \frac{f_m}{Z_{-}})}{h_0}I_{\{|f-1|> \alpha\}} \star \mu_0(\infty)\right] \nonumber \\
    &=& E\left[|f-1|(1 + \frac{f_m}{Z_{-}})I_{\{|f-1|> \alpha\}} I_{\{Z_{-}\geq\delta\}}\star \nu_0(\infty)\right] \nonumber \\
    &\leq& \delta^{-1}E\left[|f-1| I_{\{|f-1|> \alpha\}} \star (\widetilde{Z}I_{\{Z_{-}\geq \delta \}}
    \cdot \mu)(\infty)\right]<+\infty.\nonumber
  \end{eqnarray*}
  Here $\mu_0$ and $\nu_0$ are defined in (\ref{eq:cruYzero}). Therefore, again by Proposition \ref{prop:alocundergf}--(a), we conclude that $\sqrt{g_1^2\star\mu_0}\in {\cal A}_{loc}^+(\mathbb F)$.\\

  \noindent Due to $M^P_{\mu}(HK|\widetilde{\cal P}(\mathbb F))^2\leq M^P_{\mu}(H^2|\widetilde{\cal P}(\mathbb F))M^P_{\mu}(K^2|\widetilde{\cal P}(\mathbb F))$ , we derive
  \begin{eqnarray*}
  E\left[\frac{M^P_{\mu}\left(\Delta mI_{\{\widetilde{Z}>0\}} | \widetilde{\cal P}\right)^2}{Z_{-}^2h_0^2}\star \mu_0(\infty)\right] &\leq&E\left[\frac{M^P_{\mu}\left((\Delta m)^2 | \widetilde{\cal P}\right) M^P_{\mu}\left(I_{\{\widetilde{Z}>0\}} | \widetilde{\cal P}\right)}{Z_{-}^2h_0^2}\star \mu_0(\infty)\right] \nonumber \\
  &=& E\left[\frac{M^P_{\mu}\left((\Delta m)^2 | \widetilde{\cal P}\right) }{Z_{-}^2}I_{\{Z_{-}\geq \delta \}}\star \mu(\infty)\right]\\
  & \leq& \delta^{-2} E\left[ [m,m]_{\infty}\right]<+\infty. \nonumber
  \end{eqnarray*}
  Hence, we conclude that $\sqrt{(g-1)^2\star \mu_0} \in {\cal A}^+_{loc}(\mathbb{F}).$ This ends the proof of (\ref{mainassumtpionbeforetaubis}), and the proof of the theorem is completed.\qed

%%%%%%%%%%%%%%%%%%%%%%%%%%%%%%%%%%%%%%%%%%%%%%%%%%%%%%%%%%%%%%%%%%%%%%%%%%%%%%%%%%%%%%%%%%%%%%%%%%%%%%%%%%%%%%%%
%%%%%%%%%%%%%%%%%%%%%%%%%%%%%%%%%%%%%%%%%%%%%%%%%%%%%%%%%%%%%%%%%%%%%%%%%%%%%%%%%%%%%%%%%%%%%%%%%%%%%%%%%%%%%%%%%%%
%%%%%%%%%%%%%%%%%%%%%%%%%%%%%%%%%%%%%%%%%%%%%%%%%%%%%%%%%%%%%%%%%%%%%%%%%%%%%%%%%%%%%%%%
%%%%%%%%%%%%%%%%%%%%%%%%%%%%%%%%%%%%%%%%%%%%%%%%%%%%%%%%%%%%%%%%%%%%%%%%%%%%%%%%%%%%%
%%%%%%%%%%%%%%%%%%%%%%%%%%%%%%%%%%%%%%%%%%%%%%%%%%%%%%%%%%%%%%%%%%%%%%%%%%%%%%%%
%%%%%%%%%%%%%%%%%%%%%%%%%%%%%%%%%%%%%%%%%%%%%%%%%%%%%%%%%%%%%%%%%%%%%%%%%%%%%%%%%%%%%%%%%%
%%%%%%%%%%%%%%%%%%%%%%%%%%%%%%%%%%%%%%%%%%%%%%%%%%%%%%%%%%%%%%%%%%%%%%%%%%%%%%%%%%%%%%%%%%%%%%%%%%%%%%%%%%%%%%%%%%
%%%%%%%%%%%%%%%%%%%%%%%%%%%%%%%%%%%%%%%%%%%%%%%%%%%%%%%%%%%%%%%%%%%%%%%%%%%%%%%%%%%%%%%%%

  \subsection{Proof of Theorem \ref{MainTheoremMultiJumps}}\label{proofofMainTheoremMultiJumps}
  Before proving the equivalence between the two assertions of the theorem, we will start outlining a number of remarks that simplify tremendously the proof. It is easy to prove that on $\{T<+\infty\}$ we have
\begin{equation}\label{ZtildeforQandP}
\{\widetilde{Z}^Q_T=1\} = \{\widetilde{Z}_T = 1\}\ \ \mbox{for}\ Q\sim P\ \mbox{and}\ \mathbb F-\mbox{stopping time}\ T,\end{equation}
where $\widetilde{Z}^Q_t := E^Q(\tau \geq t | {\cal F}_t)$. Indeed, due to
  \begin{eqnarray*}
    E\left[(1-\widetilde{Z}_T) I_{\{\widetilde{Z}^Q_T=1\}}\right] = E\left[ I_{\{\tau< T\}} I_{\{\widetilde{Z}^Q_T=1\}}\right] = 0,
  \end{eqnarray*}
   the inclusion $\{\widetilde{Z}^Q_T=1\} \subset \{\widetilde{Z}_T = 1\}$ follows, while the reverse inclusion follows by symmetry. This proves (\ref{ZtildeforQandP}).\\
  Since $S$ is a thin process with predictable jump times only, then there exists a sequence of $\mathbb F$-predictable stopping times, $(T_n)_{n\geq 1}$, such that
  $$\{\Delta S\not=0\}\subset\bigcup_{n=1}^{+\infty}\Rbrack T_n\Lbrack.$$
The proof of the theorem consists of three steps in which we prove (b)$\Longleftrightarrow$(c), (b)$\Longrightarrow$(a) and (a)$\Longrightarrow$(b) respectively.\\
   \noindent{\bf  Step 1:} Here, we prove (b)$\Longleftrightarrow$(c). Remark that,  thanks to Lemma \ref{LY}, (c)$\Longrightarrow$(b) follows immediately. To prove the reverse (i.e. (b)$\Longrightarrow$(c)), we consider the following $\mathbb F$-predictable process
$$\varphi:= \left[1+\ ^{p,\mathbb F}\left(Y\vert \Delta S\vert I_{\{\widetilde Z<1\}}\right)\right]^{-1}\left[I_{\Omega\setminus(\bigcup_{n=1}^{+\infty} \Rbrack T_n\Lbrack)}+\sum_{n=1}^{+\infty} 2^{-n}I_{\Rbrack T_n\Lbrack}\right].$$
It is easy to check that $0<\varphi\leq 1$ and $U:=Y_{-}\varphi\cdot S^{(1)}+[Y,\varphi\cdot S^{(1)}]$ is a process with integrable variation whose compensator (since it is a pure jump process with finite variation and jumps on predictable stopping times only) is
$$U^{p,\mathbb F}=\sum \varphi\ ^{p,\mathbb F}\left(Y\Delta S I_{\{ \widetilde Z=1>Z_{-}\}}\right)\equiv 0.$$
This proves that $Y$ is $\sigma$-martingale density for $S^{(1)}$ (i.e. $Y\in {\cal L}(S^{(1)},\mathbb F)$), and hence assertion (c) follows immediately.\\
{\bf Step 2:} Here we will prove (b)$\Rightarrow $ (a). Suppose that assertion (b) holds, and consider a sequence of $\mathbb F$-stopping times $(\sigma_n)_n$ such that $Y^{\sigma_n}$ is a martingale, and put $Q_n:=(Y_{\sigma_n}/Y_0)\cdot P\sim P$. Then, since $\widetilde S:=\sum \Delta S I_{\{\widetilde Z<1\}}$ is a thin process with predictable jump times only, the condition (\ref{mainEquationMultiJumps}) translates into the fact that $ {\widetilde S}^{\sigma_n}$ is a $Q_n$-local martingale satisfying $$\{\Delta \widetilde S^{\sigma_n}\not=0\}\cap\{\widetilde Z^{Q_n}=1>Z^{Q_n}_{-}\}=\emptyset, $$
due to (\ref{ZtildeforQandP}). Therefore, thanks to Proposition \ref{NUPBRLocalization}, it is enough to prove that assertion (a) holds true under $Q_n$ for $S^{\sigma_n}$. Therefore, without loss of generality, we assume $Y\equiv 1$ and hence $\widetilde S$ is a $\mathbb F$-local martingale satisfying (\ref{Condition4AfterTau}). Thus, a direct application of Theorem \ref{explicitdeflatorThin4AfterTau} implies that ${\widetilde S}^{\tau}$ satisfies the NUPBR$(\mathbb G)$.\\
\noindent {\bf Step 3:} Here, we will prove (a)$\Rightarrow $(b). Suppose that $S -
S^\tau$ satisfies the NUPBR($\mathbb G$). A direct application Theorem
\ref{theosigmadensityiff} implies the existence of $f^{\mathbb
G}\in {\cal G}^1_{loc}(\mu^{\mathbb G}_a,\mathbb G)$ such that $f^{\mathbb G}>0$,
  \begin{eqnarray*}
    N^{\mathbb G}&:=&W^{\mathbb G}\star (\mu^{\mathbb G}_a - \nu^{\mathbb G}_a),\ \  W^{\mathbb G}:= f^{\mathbb G}-1 + \frac{\widehat{f}^{\mathbb G} - a^{\mathbb G}}{1 - a^{\mathbb G}}I_{\{a^{\mathbb G} <1\}},\end{eqnarray*}
   and
  \begin{eqnarray}\label{mgunderGaftertau}
    xf^{\mathbb G}\star \nu^{\mathbb G}_a  = xf^{\mathbb G}\left(1 - \frac{f_m}{1-Z_{-}}\right)I_{\Lbrack \tau, +\infty\Rbrack}\star  \nu \equiv 0.
  \end{eqnarray}
  Here $f_m := M^P_{\mu}(\Delta m | \widetilde{\cal P}(\mathbb F))$ (given also by (\ref{decompositionofm}), and $\mu^{\mathbb G}_a$ and $\nu^{\mathbb G}_a$ are given by $$\mu^{\mathbb G}_a := I_{\Lbrack \tau, +\infty\Rbrack}\cdot \mu,\ \nu^{\mathbb G}_a := \left(1 - \frac{f_m}{1 - Z_{-}}\right)I_{\Lbrack \tau, +\infty\Rbrack}\cdot \nu.$$
   Thanks to Lemma \ref{lemma:predsetFG}, there exists an ${\cal P}(\mathbb F)$-measure functional $f>0$ such that $f^{\mathbb G} = f$ on the stochastic interval $\Lbrack \tau, +\infty\Rbrack$, and (\ref{mgunderGaftertau}) becomes
  \begin{eqnarray}\label{eq:aftertaufmcru}
    xf\left(1 - \frac{f_m}{1 -Z_{-}}\right)I_{\Lbrack \tau, +\infty\Rbrack}\star  \nu \equiv 0.
  \end{eqnarray}
  Due to Proposition \ref{Gstoppingtimeaftertau} and $\mathbb G$-locally boundedness of $(1-Z_{-})^{-1}I_{\Lbrack \tau, +\infty\Rbrack}$, we could find a sequence of $\mathbb F$-stopping time $(\sigma_n^{\mathbb F})_{n\geq 1}$ that increases to infinity and  $(1-Z_{-})^{-1}I_{\Rbrack 0,\sigma^{\mathbb F}_n\Lbrack}I_{\Lbrack \tau, +\infty\Rbrack} $ is bounded by $(n+1)$. Also, since $\left((f-1)^2 I_{\Lbrack \tau, +\infty\Rbrack}\star \mu \right)^{1/2}  \in {\cal A}^+_{loc}(\mathbb{G})$, thanks to Proposition \ref{Gstoppingtimeaftertau} (both assertions (c) and (a)) we deduce the existence of a sequence of $\mathbb F$-stopping times $(\tau_n)_n$ that increases to infinity such that the three processes $[m,m]^{\tau_n}$
  \begin{eqnarray*}
   (f-1)^2 I_{\{|f-1|\leq \alpha\ \&\ 1-Z_{-}\geq\delta\}}\star\overline{\mu})^{\tau_n}\ \mbox{and}\  |f-1|I_{\{|f-1|> \alpha \&\ 1-Z_{-}\geq\delta\}}\star{\overline \mu})^{\tau_n}
  \end{eqnarray*}
  are integrable, where $\overline \mu:=(1-\widetilde{Z})\cdot \mu$. Consider the following notations
  \begin{eqnarray*}\label{eq:cruYzeroafter}
   \mu_1 &:=& I_{\{\widetilde{Z}<1\ \&\ 1-Z_{-}\geq\delta\}} \cdot \mu , \ \ \nu_1 := h_1 I_{\{ 1-Z_{-}\geq\delta\}}\cdot \nu, \   h_1:= M^P_{\mu}\left(I_{\{\widetilde{Z}<1\}} | \widetilde{\cal P}\right), \nonumber \\
    g &:=& \frac{f(1 - \frac{f_m}{1-Z_{-}})}{h_1} I_{\left\{h_1>0\ \&\ Z_{-}<1\right\}} + I_{\left\{h_1=0\ \mbox{or}\ Z_{-}=1\right\}},
  \end{eqnarray*}
  and suppose that
  \begin{equation}\label{mainassumption3.4}
  W^{(1)}(t,x):={{g_t(x)-1}\over{1-a^{(1)}_t+\widehat g_t}}\in {\cal G}^1_{loc}(\mu_1,\mathbb F),
  \end{equation}
  where $ a^{(1)}_t:=\nu_1(\{t\},\mathbb R^d)$ and $\widehat g_t:=\int g_t(x)\nu_1(\{t\},dx).$

  \noindent Then, we can easily prove that assertion (b) holds. In fact, we take
  $$
   N^{(1)}:= \frac{g-1}{1 - a^{(1)} + \widehat{g}}\star(\mu_1 - \nu_1)\ \ \mbox{and}\ \  \ Y := {\cal E}(N^{(1)}).$$
  Then, it is clear that $$1 + \Delta N^{(1)} = {1\over{1 - a^{(1)} + \widehat{g}}}I_{\{\Delta S=0\ \mbox{or}\ \widetilde Z=1\}}+\frac{g(\Delta S)}{1 - a^{(1)} + \widehat{g}} I_{\{ \Delta S\not=0\ \&\ \widetilde Z<1\}}> 0,$$ and on $\{Z_{-}<1\}$ we get
  \begin{eqnarray*}
    ^{p,\mathbb F}\left(Y\Delta S I_{\{\widetilde{Z}<1\}}\right)_t &=& Y_{t-}\ ^{p,\mathbb F}\left((1+\Delta N^{(1)})\Delta S I_{\{\widetilde{Z}<1\}}\right)_t =\frac{Y_{t-}\  ^{p,\mathbb F}\left(g(\Delta S)\Delta S I_{\{\widetilde{Z}<1\}}\right)_t}{1 - a^{(1)}_t + \widehat{g}_t}\\
     &=& \frac{Y_{t-}}{1 - a^{(1)}_t + \widehat{g}_t}\int g_t(x)xh_1(t,x) \nu(\{t\},dx) \\
     &=& \frac{Y_{t-}}{1 - a^{(1)}_t + \widehat{g}_t} \int xf_t(x)\left(1 - \frac{f_m(t,x)}{1-Z_{-}}\right)\nu(\{t\},dx) \equiv 0.
  \end{eqnarray*}
  The last equality in the above string of equalities follows direct from (\ref{eq:aftertaufmcru}). Therefore, assertion (b) will follow immediately as long as we prove (\ref{mainassumption3.4}). To this end, on $\left\{h_1>0\ \&\ Z_{-}<1\right\}$ we calculate
  \begin{eqnarray*}\label{gdecomposition}
    g-1 &=& \frac{f(1-Z_{-}- f_m)}{h_1(1-Z_{-})} - 1= \frac{(f-1)(1-Z_{-}- f_m)}{h_1(1-Z_{-})}-\frac{M^P_{\mu}\left(\Delta mI_{\{\widetilde{Z}<1\}} | \widetilde{\cal P}\right)}{(1-Z_{-})h_1}\\
    &=&:g_1 + g_2,
  \end{eqnarray*}
  and remark that $\{ 1-Z_{-}-f_m>0\}\subset\{h_0>0\}$ which is due to $$1-Z_{-}-f_m=1 - M^P_{\mu}\left(\widetilde{Z}| \widetilde{\cal P}\right) \leq M^P_{\mu}\left(I_{\{\widetilde{Z}<1\}} | \widetilde{\cal P}\right) = h_1,$$ that is implied by $I_{\{\widetilde{Z}=1\}} \leq \widetilde{Z}$. Therefore, we derive that
  \begin{eqnarray*}
    E\left[g_1^2I_{\{|f-1|\leq \alpha\}}\star \mu_0(\sigma_n\wedge\tau_n)\right] &=& E\left[\frac{(f-1)^2(1 -Z_{-}- f_m)^2}{h_1^2(1 -Z_{-})^2}I_{\{|f-1|\leq \alpha\}} I_{\{\widetilde{Z}<1\}}\star \mu(\sigma_n\wedge\tau_n)\right] \nonumber \\
    &=& E\left[\frac{(f-1)^2(1 -Z_{-}- f_m)^2}{h_1(1 -Z_{-})^2}I_{\{|f-1|\leq \alpha\}} \star \nu(\sigma_n\wedge\tau_n)\right] \nonumber \\
    &\leq& E\left[(f-1)^2 \frac{1-Z_{-} - f_m}{(1-Z_{-})^2} I_{\{|f-1|\leq \alpha\}} \star \nu(\sigma_n\wedge\tau_n)\right] \nonumber \\
    &\leq& E\left[(f-1)^2\frac{1}{(1-Z_{-})^2} I_{\{|f-1|\leq \alpha\}} \star ((1-\widetilde{Z})\cdot \mu(\sigma_n\wedge\tau_n))\right] \nonumber\\
    &=& E\left[(f-1)^2\frac{1}{(1-Z_{-})^2} I_{\{|f-1|\leq \alpha\}}I_{\Lbrack \tau,+\infty\Rbrack} \star  \mu(\sigma_n\wedge\tau_n)\right]\nonumber\\
    &\leq& (n+1)^2 E\left[(f-1)^2I_{\{|f-1|\leq \alpha\}} \star  \mu(\tau_n)\right]<+\infty.
  \end{eqnarray*}
  and
  \begin{eqnarray*}
    E\left[g_1I_{\{|f-1|> \alpha\}}\star \mu_1(\sigma_n\wedge\tau_n)\right] &=& E\left[\frac{|f-1|(1-Z_{-}- f_m)}{h_1(1-Z_{-})}I_{\{|f-1|> \alpha\}} I_{\{\widetilde{Z}<1\}}\star \mu(\sigma_n\wedge\tau_n)\right] \nonumber \\
    &=& E\left[|f-1|\frac{1-Z_{-}-f_m}{1-Z_{-}}I_{\{|f-1|> \alpha\}} \star \nu(\sigma_n\wedge\tau_n)\right] \nonumber \\
    &\leq& E\left[|f-1| \frac{1}{1-Z_{-}} I_{\{|f-1|> \alpha\}} \star ((1-\widetilde{Z})\cdot \mu)(\sigma_n\wedge\tau_n)\right] \nonumber \\
    &=& E\left[|f-1| \frac{1}{1-Z_{-}} I_{\{|f-1|> \alpha\}}I_{\Lbrack \tau,+\infty\Rbrack} \star  \mu(\sigma_n\wedge\tau_n)\right] \nonumber \\
    &\leq& (n+1)E\left[|f-1| I_{\{|f-1|> \alpha\}}\star  ((1-\widetilde Z)\cdot \mu)(\tau_n)\right] <+\infty.
  \end{eqnarray*}
  This proves that $\sqrt{g_1^2\star \mu_1}\leq \sqrt{2}\sqrt{g_1^2 I_{\{ \vert f-1\vert>\alpha\}}\star \mu_1}+\sqrt{2}\left(g_1 I_{\{ \vert f-1\vert>\alpha\}}\star \mu_1\right)$ belongs to ${\cal A}^+_{loc}(\mathbb F)$. To prove $\sqrt{g_2^2\star \mu_1}\in {\cal A}^+_{loc}(\mathbb F)$, we derive
  \begin{eqnarray*}
  E\left[(g_2)^2\star \mu_0(\sigma_n\wedge\tau_n)\right] &=&E\left[\frac{M^P_{\mu}\left(\Delta mI_{\{\widetilde{Z}<1\}} | \widetilde{\cal P}\right)^2}{(1-Z_{-})^2h_0^2}I_{\{\widetilde{Z}<1\ \&\ Z_{-}<1\}}\star \mu(\sigma_n\wedge\tau_n)\right] \nonumber \\
  &\leq&E\left[\frac{M^P_{\mu}\left((\Delta m)^2 | \widetilde{\cal P}\right) M^P_{\mu}\left(I_{\{\widetilde{Z}<1\ \&\ Z_{-}<1\}} | \widetilde{\cal P}\right)^2}{(1-Z_{-})^2h_1^2}\star \nu(\sigma_n\wedge\tau_n)\right] \nonumber \\
  &=& E\left[\frac{M^P_{\mu}\left((\Delta m)^2 | \widetilde{\cal P}\right) }{(1-Z_{-})^2}I_{\{Z_{-}<1\}}\star \nu(\sigma_n\wedge\tau_n)\right] \nonumber \\
  &\leq&  E\left[\frac{1}{(1-Z_{-})^3}I_{\Lbrack \tau,+\infty\Rbrack}\is\left(M^P_{\mu}\left((\Delta m)^2 | \widetilde{\cal P}\right)\star\nu\right)_{\sigma_n\wedge\tau_n}\right]\nonumber\\
   &\leq&(n+1)^3 E([m,m]_{\tau_n}) <+\infty.
  \end{eqnarray*}
  Hence, $\sqrt{(g-1)^2\star \mu_1} \in {\cal A}_{loc}^+(\mathbb{F})$ follows. Thanks to Lemma \ref{boundednessofuhat} (see Choulli and Schweizer(2012) \cite{Choulli2012}), (\ref{mainassumption3.4}) follows immediately. This ends the proof of the theorem. \qed

\appendix

\section{Integrality Results}
\begin{theorem}\label{theosigmadensityiff}
  Let $S$ be a semi-martingale with predictable characteristic triplet $(b,c,\nu=A\otimes F)$, $N$ is a local martingale such that ${\cal E}(N)>0$, and $(\beta,f,g,N')$ are its Jacod's parameters.  Then the following assertions hold.\\
  1) ${\cal E}(N)$ is a $\sigma$-martingale density of $S$  if and only if the following two properties hold:
  \begin{equation}\label{integrabilitycondition}
  \displaystyle\int \vert x - h(x) + xf(x)\vert F(dx)<+\infty,\ \ \ P\otimes A-a.e.\end{equation}
  and
  \begin{equation}\label{martingalerequation}
  b + c\beta + \displaystyle\int\Bigl(x - h(x) + x f(x)\Bigr)F(dx) = 0,\ \ \ \ \ P\otimes A-a.e.\end{equation}
\end{theorem}
2) In particular, we have
 \begin{equation}\label{martingalerequationJumps}
 \int x(1+f_t(x)\nu(\{t\},dx) = \int x(1+f_t(x)F_t(dx)\Delta A_t=0,\ \ \ \ \ P-a.e.\end{equation}
\begin{proof}
  The proof can be found in Choulli et al. \cite[Lemma 2.4]{choullistricker07} 2007, and also Choulli and Schweizer (2013) \cite{Choulli2012}.\qed
\end{proof}

\begin{lemma}\label{boundednessofuhat} Let $f$ be a $\widetilde{\cal P}(\mathbb H)$-measurable functional such that $f>0$  and
\begin{equation}\label{(f-1)G1loc}\Bigl[(f-1)^2\star\mu\Bigr]^{1/2}\in {\cal A}^+_{loc}(\mathbb  H).\end{equation}
Then, the $\mathbb H$-predictable process $\left(1-a^{\mathbb H}+{\widehat f}^{\mathbb H}\right)^{-1}$ is locally bounded, and hence
\begin{equation}\label{WinG1loc}
W_t(x):={{f_t(x)-1}\over{1-a^{\mathbb H}_t+{\widehat f}^{\mathbb H}_t}}\in {\cal G}^1_{loc}(\mu,\mathbb H).\end{equation}
Here, $a^{\mathbb H}_t:=\nu^{\mathbb H}(\{t\},\mathbb R^d)$, ${\widehat f}^{\mathbb H}_t:=\int f_t(x)\nu^{\mathbb H}(\{t\},dx)$ and $\nu^{\mathbb H}$ is the $\mathbb H$-predictable random measure compensator of $\mu$ under $\mathbb H$.
\end{lemma}

\begin{proof} The proof of this lemma can be found in Choulli and Schweizer (2013). Below we provide this proof for the sake of completeness. In this proof we put
$$
U_t(x)=1-f_t(x),\ \ \ \mbox{and}\ \ \ \widehat U_t:=a^{\mathbb H}_t-{\widehat f}^{\mathbb H}_t.$$
 We start by remarking that (\ref{WinG1loc}) follows from the combination of (\ref{(f-1)G1loc}) and the local boundedness of $1/(1-\widehat U)$. Therefore, in what follows, we will focus on proving this latter fact. Consider $\delta\in (0,1)$, $\eta\in (0,1)$, and the stopping times and processes defined by
$$\begin{array}{llll}
T_0=0,\ \ \ T_{n+1}:=\displaystyle\inf\left\{t>T_n\ \big|\  \sum_{T_n<v\leq t}\left(U_v(\Delta S_v)I_{\{ \Delta S_v\not=0\}}\right)^2>\delta^2\ \right\},\\
\\
V_n(t):=\Bigl[ \displaystyle\sum_{T_n<v\leq t}\left(U_v(\Delta S_v)I_{\{ \Delta S_v\not=0\}}\right)^2\Bigr]^{1/2}\end{array}.$$
Remark that ---since for each $n\geq 0$, the process $(V_n(t))^2$ is RCLL and nondecreasing real-valued process---- we have
$$
(V_n(T_{n+1}))^2:=\displaystyle \sum_{T_n<v\leq T_{n+1}}\left(U_v(\Delta S_v)I_{\{ \Delta S_v\not=0\}}\right)^2\geq\delta^2\ \ \mbox{on}\ \ \{T_{n+1}<+\infty\}.$$
This implies that $T_n$ increases to $+\infty$ almost surely, and
$$
V_n(t-)\leq \delta,\ \ P-a.s.\ \ \ \ \mbox{for all}\ t\leq T_{n+1}.$$
Due to $0\leq (1-\widehat U)^{-1}I_{\{ \widehat U< 1-\eta\}}\leq \eta^{-1}$ and
$$
(1-\widehat U)^{-1}=(1-\widehat U)^{-1}I_{\{ \widehat U\geq 1-\eta\}}+(1-\widehat U)^{-1}I_{\{ \widehat U< 1-\eta\}} , $$
we deduce that the proof of the lemma will achieved once we prove that
$$Y:={1\over{1-\widehat U}}I_{\{ \widehat U\geq 1-\eta\}}$$ is locally bounded.  Thanks to Dellacherie and Meyer (1980), this fact is equivalent to
$$
\sup_{0\leq u\leq t} Y_u<+\infty\ \ P-a.s.\ \ \  \mbox{for any}\ t\in (0,+\infty).$$
Since $T_n$ increases to $\infty$ almost surely, then this fact is implied by
$$
\sup_{T_n\leq u\leq t\wedge T_{n+1}} Y_u<+\infty\ \ P-a.s.\ \ \ \mbox{on}\ \ \{t>T_n\}.$$
Simple calculation leads to
$$
\widehat U_s\leq V_n(s-)+\ ^{p,\mathbb H}(\Delta V_n)_s,\ \ \ \mbox{for all}\ \ T_n<s\leq T_{n+1}. $$ Thus, it is easy to see that for $\delta+\eta<1$,
$$\begin{array}{llll}
\{s\in ]T_n,T_{n+1}]\ \big| \ \widehat U_s\geq 1-\eta\}\subset\{ s\in ]T_n,T_{n+1}]\ \big| \ ^{p,\mathbb H}(\Delta V_n)_s\geq 1-\eta-V_n(s-)\}\\
\\
\hskip 3cm \subset \{ s\in ]T_n,T_{n+1}]\ \big| \ \Delta\left((V_n)^{p,\mathbb H}\right)=\ ^{p,\mathbb H}(\Delta V_n)_s\geq 1-\eta-\delta\}=:\Gamma_n.\end{array}$$
It is obvious that $\#\left(\Gamma_n\cap [0,t]\right)<+\infty\ P-a.s.$ since $(V_n)^{p,\mathbb H}$ is a c\`adl\`ag process. Thus, we deduce that
$$
\sup_{T_n\leq u\leq t\wedge T_{n+1}} Y_u=\max_{T_n\leq u\leq t\wedge T_{n+1}} Y_u<+\infty.$$
This ends the proof of the lemma.\qed\end{proof}

\begin{proposition}\label{prop:alocundergf}
  For any $\alpha>0$, the following assertions hold:\\
  {\rm{(a)}} Let $h$ be a ${\widetilde{\cal P}(\mathbb H)}$-measurable functional. Then, $\sqrt{(h-1)^2\star \mu} \in {\cal A}^+_{loc}(\mathbb H)$ iff
  \begin{eqnarray*}
    (h-1)^2I_{\{|h-1|\leq \alpha\}}\star \mu \mbox{ and } \ |h-1|I_{\{|h-1|>\alpha\}}\star \mu\ \ \mbox{belong to}\ \ {\cal A}^+_{loc}(\mathbb H).
  \end{eqnarray*}
   {\rm{(b)}} Let $V$ be an $\mathbb F$-predictable and non-decreasing process.
   Then, $V^{\tau}\in{\cal A}_{loc}^+(\mathbb G)$ if and only if $I_{\{ Z_{-}\geq\delta\}}
   \is V\in {\cal A}_{loc}^+(\mathbb F)$ for any $\delta>0$.\\
 {\rm{(c)}}
 Let $h$ be a nonnegative and $\widetilde{\cal P}(\mathbb F)$-measurable functional.
 Then, $hI_{\Lbrack 0,\tau \Lbrack}\star \mu \in {\cal A}^+_{loc}(\mathbb G)$ if and
 only if for all $\delta >0$,
  $hI_{\{Z_{-}\geq \delta \}}\star \mu^1 \in {\cal A}^+_{loc}(\mathbb  F)$,
  where $\mu^1:=\widetilde{Z}\centerdot\mu.$\\
  {\rm{(d)}} Let $f$ be positive and $\widetilde{\cal P}(\mathbb F)$-measurable, and $\mu^1:=\widetilde{Z}\centerdot\mu.$ Then $\sqrt{(f-1)^2I_{\Lbrack 0,\tau \Lbrack}\star \mu} \in {\cal A}^+_{loc}(\mathbb G)$ iff $\sqrt{(f-1)^2I_{\{Z_{-}\geq \delta\}}\star \mu^1} \in {\cal A}^+_{loc}(\mathbb F)$, for all $\delta >0.$
\end{proposition}

\begin{proposition}\label{Gstoppingtimeaftertau}
  Suppose that  $\tau$ is a finite honest time satisfying (\ref{mainassumptionontau}). Then, the following properties hold.\\
   {\rm{(a)}} Let $\Phi^{\mathbb G}$ a $\mathbb G$-predictable process and $k$ a nonnegative and $\widetilde{\cal P}(\mathbb F)$-measurable functional such that $0<\Phi^{\mathbb G}\leq 1 $ and $\Phi^{\mathbb G}k\star\mu_{\mathbb G}\in {\cal A}^+_{loc}(\mathbb G)$. Then, $P\otimes A$-a.e.
   \begin{equation}\label{finitenessmuG}
   \int k(x)\left(1-Z_{-}-f_m(x)\right)F(dx)<+\infty\ \ \ \ \mbox{on}\ \ \ \{Z_{-}<1\}.\end{equation}
  {\rm{(b)}} Let $f$ be a $\widetilde{\cal P}(\mathbb F)$-measurable and positive functional, and ${\overline\mu}:=(1-\widetilde{Z})\cdot \mu.$ Then $\sqrt{(f-1)^2I_{\Lbrack \tau,+\infty \Rbrack}\star \mu} \in {\cal A}^+_{loc}(\mathbb G)$ if and only if $\sqrt{(f-1)^2 I_{\{ 1-Z_{-}\geq\delta\}}\star {\overline\mu} }\in {\cal A}^+_{loc}(\mathbb F)$ for any $\delta>0$.
\end{proposition}

\section{Representation Results}

\begin{lemma}\label{lemma:predsetFG} The  following assertions
 hold.\\
 {\rm{(a)}} If $H^{\mathbb G}$ is a $\widetilde{\cal P}(\mathbb G)$-measurable functional, then there exist an $\widetilde{\cal
P}(\mathbb F)$-measurable functional $H^{\mathbb F}$such that
\begin{eqnarray}\label{eq:widePGandwidePF}
H^{\mathbb G}(\omega,t,x)I_{\Lbrack 0,\tau\Lbrack} = H^{\mathbb F}(\omega,t,x) I_{\Lbrack 0,\tau\Lbrack}.
\end{eqnarray}
 {\rm{(b)}} If furthermore $H^{\mathbb G}>0$ (respectively $H^{\mathbb G}\leq 1$), then we can choose $H^{\mathbb F}>0$ (respectively $H^{\mathbb F}\leq 1$) such that $$H^{\mathbb G}(\omega,t,x)I_{\Lbrack 0,\tau\Lbrack} = H^{\mathbb F}(\omega,t,x) I_{\Lbrack 0,\tau\Lbrack}.$$
( {\rm{(c)}} For any $\mathbb F$-stopping time, $T$, and any positive  ${\cal G}_T$-measurable random variable $Y^{\mathbb G}$, there exist two positive ${\cal F}_{T}$-measurable random variables, $Y^{(1)}$ and $Y^{(2)}$, satisfying
\begin{equation}\label{splitatT}
 Y^{\mathbb G}I_{\{ T\leq \tau\}}=Y^{(1)}I_{\{T<\tau\}}+Y^{(2)}I_{\{\tau=T\}}.\end{equation}
\end{lemma}

\begin{lemma}\label{lemma:predsetFG} Suppose that $\tau$ is honest. Let $H^{\mathbb G}$ be an $\widetilde{\cal P}(\mathbb
G)$-measurable functional. Then the following assertions
 hold.\\
{\rm{(a)}}  There exist two $\widetilde{\cal
P}(\mathbb F)$-measurable functional $H^{\mathbb F}$ and $K^{\mathbb F}$ such
that
\begin{eqnarray}\label{eq:widePGandwidePF}
H^{\mathbb G}(\omega,t,x) = H^{\mathbb F}(\omega,t,x) I_{\Lbrack
0,\tau\Lbrack}+ K^{\mathbb F}(\omega,t,x)I_{\Lbrack
\tau,+\infty\Rbrack}.
\end{eqnarray}
{\rm{(b)}}  If furthermore $H^{\mathbb G}>0$ (respectively $H^{\mathbb
G}\leq 1$), then we can choose $K^{\mathbb F}>0$ (respectively
$K^{\mathbb F}\leq 1$) such that $$H^{\mathbb
G}(\omega,t,x)I_{\Lbrack\tau,+\infty\Lbrack} = K^{\mathbb
F}(\omega,t,x) I_{\Lbrack\tau,+\infty\Rbrack}.$$
\end{lemma}

%%%%%%%%%%%%%%%%%%%%%%%%%%%%%%%%%%%%%%%%%%%%%%%%%%%%%%%%%%%%%%%%%%%%%%%%%%%%%%%%%%%%%%%%%%%%%%%%%%%%%%%%%%%%%%%%%%%%%%%%%
%%%%%%%%%%%%%
%%%%%%%%%%%%%%%%%%%%%%%%%%%%%%%%%%%%%%%%%%%%%%%%%%%%%%%%%%%%%%%%%%%%%%%%%%%%%%%%%%%%%%%%%%%%%%%%%%%%%%%%%%%%%%%%%%%%%%%%%%%%%%
%%%%%%%%%%%%%%%%%%%%%%%%%%%%%%%%%%%%%%%%%%%%%%%%%%%%%%%%%%%%%%%%%%%%%%%%%%%%%%%%%%%%%%%%%%%%%%%%%%%%%%%%%%%%%%%%%%%%%%%%%%%
%%%%%%%%%%%%%%%%%%%%%%%%%%%%%%%%%%%%%%%%%%%%%%%%%%%%%%%%%%%%%%%%%%%%%%%%%%%%%%%%%%%%%%%%%%%%%%%%%%%%%%%%%%%%%%%%%%%%%%%%%%%%
%%%%%------------------------------------------------------------------------------------------------------------------
%%%%%%%%%%%%%%%%%%--------------------------------------------------------------------------------------------------------------------------
%%%%%%%%%%%%%%%%--------------------------------------------------------------------------------------------------------------------------
%%%%%%%%%%%%%%%%%%%%%%%%%%%%%%%%%%%%%%%%%%%%%%%%%%%%%%%%%%%%%%%%%%%%%%%%%%%%%%%%%%%%%%%%%%%%%%%%%%%%%%%%%%%%%%%%%%%%%%%%%%
%%%%%%%%%%%%%%%%%%%%%%%%%%%%%%%%%%%%%%%%%%%%%%%%%%%%%%%%%%%%%%%%%%%%%%%%%%%%%%%%%%%%%%%%%%%%%%%%%%%%%%%%%%%%%%%%%%%%%%%%%%%
%%%%%%%%%%%%%%%%%%%%%%%%%%%%%%%%%%%%%%%%%%%%%%%%%%%%%%%%%%%%%%%%%%%%%%%%%%%%%%%%%%%%%%%%%%%%%%%%%%%%%%%%%%%%%%%%%%%%%%
%%%%%%%%%%%%%%%%%%%%%%%%%%%%%%%%%%%%%%%%%%%%%%%%%%%%%%%%%%%%%%%%%%%%%%%%%%%%%%%%%%%%%%%%%%%%%%%%%%%%%%%%%%%%%%%%%%%%%%%%%%

\section*{Acknowledgements}
The research of Tahir Choulli  and Jun Deng is supported financially by the
Natural Sciences and Engineering Research Council of Canada,
through Grant G121210818. The research of Anna Aksamit and Monique Jeanblanc is supported
by Chaire Markets in transition, French Banking Federation.

%%%%%%%%%%%%%%%%%%%%%%%%%%%%%%%%%%%%%%%%%%%%%%%%%%%%%%%%%%%%%%%%%%%%%%%%%%%%%%%%%%
%%%%%%%%%%%%%%%%%%%%%%%%%%%%%%%%%%%%%%%%%%%%%%%%%%%%%%%%%%%%%%%%%%%%%%%%%%%%%%%%%
%\appendix
%\section{$\sigma$-martingale Densities in Terms Charateristics}
%%%%%%%%%%%%%%%%%%%%%%%%%%%%%%%%%%%%%%%%%%%%%%%%%%%%%%%%%%%%%%%%%%%%%%%%%%%%%%%%%

%{\bf References}

%\newpage

\end{document}